\documentclass[11pt]{article}

\usepackage{amsthm,amsmath,amssymb,bbm}
\usepackage{natbib}
\usepackage{multirow}
\usepackage[pdftex]{graphicx}
\usepackage{subfigure}
\usepackage{makecell}
\usepackage{booktabs}
\usepackage{array}
\usepackage{tabularx}
\usepackage{caption}
\usepackage{booktabs}
\usepackage{url}
\usepackage{algorithm}
\usepackage{algorithmic}
\usepackage{bm}
\usepackage{lipsum}
\usepackage{mathrsfs}
\usepackage{dsfont}
\usepackage{txfonts}

\usepackage[usenames,dvipsnames,svgnames,table]{xcolor}
\usepackage[colorlinks,
linkcolor=red,
anchorcolor=blue,
citecolor=blue
]{hyperref}

\usepackage{smile}

\newcommand*{\var}{\textnormal{var}}

\newcommand{\nn}{\nonumber}

\def\##1\#{\begin{align}#1\end{align}}
\def\$#1\${\begin{align*}#1\end{align*}}

\usepackage{relsize}
\newcommand{\T}{{\mathsmaller {\rm T}}}

\def\sn{\sum_{i=1}^n}

\newcommand{\pr}{\mathbb{P}}

\newcommand{\loc}{{{\rm loc}}}
\newcommand{\lsc}{{{\rm lsc}}}
\newcommand{\imd}{{\mathsmaller {\rm imd}}}

\def\sn{\sum_{i=1}^n}

\newcommand{\wt}{\widetilde}
\newcommand{\wh}{\widehat}
\newcommand{\bfsym}[1]{\ensuremath{\boldsymbol{#1}}}
       \def \bbeta    {\bfsym{\beta}}

\newcommand{\Rom}[1]{\text{\uppercase\expandafter{\romannumeral #1\relax}}}

\usepackage{geometry}
\usepackage[shortlabels]{enumitem}

\numberwithin{equation}{section}

\begin{document}

\title{Distributed Adaptive Huber Regression}

\author{Jiyu Luo\thanks{Department of Family Medicine and Public Health, University of California, San
Diego, La Jolla, CA 92093, USA. E-mail: \href{mailto:jil130@ucsd.edu}{\textsf{jil130@ucsd.edu}}.},~~~Qiang Sun\thanks{Department of Statistical Sciences, University of Toronto, Toronto, ON M5S 3G3, Canada. E-mail: \href{mailto:qsun@utstat.toronto.edu}{\textsf{qsun@utstat.toronto.edu}}.}~~~and~~Wen-Xin Zhou\thanks{Department of Mathematics, University of California, San Diego, La Jolla, CA 92093, USA. E-mail:  \href{mailto:wez243@ucsd.edu}{\textsf{wez243@ucsd.edu}}.} }
\date{}

\maketitle

\begin{abstract}
Distributed data naturally arise in scenarios involving multiple sources of observations,  each  stored at a different location. Directly pooling  all the data together is often prohibited due to limited bandwidth and storage, or due to  privacy protocols. This paper introduces a new robust distributed algorithm for fitting linear regressions when data are subject to heavy-tailed and/or asymmetric errors with finite second moments. The algorithm only communicates gradient information at each iteration, and therefore is communication-efficient. Statistically, the resulting estimator achieves the centralized nonasymptotic error bound as if all the data were pooled together and came from a distribution with sub-Gaussian tails. Under a finite $(2+\delta)$-th moment condition, we derive a Berry-Esseen bound for the distributed estimator, based on which we construct  robust confidence intervals.  Numerical studies further confirm that compared with extant distributed methods, the proposed methods achieve near-optimal accuracy with low variability and better coverage with tighter confidence width.
\end{abstract}


\section{Introduction}
\label{sec:1} 
 
In many applications, there are a massive number of individual agents/organizations collecting data independently. Multiple-site research has brought the possibility  of  studying rare outcome that require  larger sample sizes, accelerating more generalizable findings, and bringing together investigators with different expertise from various backgrounds \citep{sidransky2009multicenter}. 
Due to limited resources, such as bandwidth and storage, or privacy concerns, researchers across different sites are only allowed to share summary statistics without allowing collaborating parties to access raw data \citep{W2012}.
Moreover, the collected data may often be contaminated by high level of noise,  and thus of low quality. For example, in the  context of gene expression data analysis, it has been observed  that some gene expression levels have kurtosis values much larger than $3$, despite of the normalization methods used \citep{wang2015high}. It is therefore important to develop robust and distributed learning algorithms with controlled communication cost and desirable statistical performance, measured by both efficiency and robustness.

Distributed learning algorithms have received considerable attention for multi-source studies in the past decade. Due to privacy concerns, data collected at each source, such as node, sensor or organization, must remain local. The goal is to develop efficient statistical learning methods that allow shared analyses or summary statistics without sharing individual level data.
The classical divide-and-conquer principle is based on aggregating local estimators, that is, estimators computed separately on local machines, to form a final estimator; see, for example, \cite{chen2012split},  \cite{LLL2013}, \cite{ZDW2015}, \cite{ZCL2016},  \cite{RN2016},  \cite{LLST2017},  \cite{BFLLZ2018} and \cite{VCC2019},  among many others. We refer to \cite{HC2018} for a more complete literature review.
One-step averaging takes one communication round, and therefore is convenient and has minimal communication cost.    However,  in order for the averaging estimator to achieve to same convergence rate as the centralized estimator, each local machine must have access to at least $\sqrt{N}$ samples, where $N$ is the total sample size.  This limits the number of machines allowed in the communication network.

To overcome this barrier of one-step averaging, multi-round procedures have been proposed for distributed data analysis with a large number of local agents \citep{SSZ2014, WKSZ2017,JLY2018,wang2019distributed}.  For linear and generalized linear models,  \cite{WKSZ2017} and \cite{JLY2018}  proposed multi-round distributed (penalized) $M$-estimators that achieve optimal rates of convergence under very mild constraints on the number of machines.  \cite{CLZ2019} studied an iterative algorithm with proper smoothing for quantile regression under memory constraint, which may also apply under distributed computing platform.
 Alternatively, \cite{DS2018} proposed an iterative weighted parameter averaging scheme for distributed linear regression when the dimension is comparable to the sample size.
 
For linear models under data parallelism, most of the existing distributed algorithms work with the least squares method, either by (weighted) averaging local least squares estimators or iteratively minimizing shifted (penalized) least squares loss functions. From a robustness viewpoint,  distributed least squares based method inherits the sensitivity (non-robustness) of its centralized counterpart to  the tails of the error distributions, hence increasing the variability of the estimator.
 In this paper, we propose a robust distributed algorithm for linear regression with heavy-tailed errors.
Our setup includes the heteroscedastic linear model with asymmetric errors,  to which the least absolute deviation (LAD) regression does not naturally apply.   Following the terminology in \cite{C2012}, the type of ``robustness" considered in this paper is quantified by nonasymptotic exponential deviation of the estimator versus polynomial tail of the error distribution.  The ensuing procedure does sacrifice a fair amount of robustness to adversarial contamination of the data.  The motivation of this work is different from and should not be confused with the classical notion of robust statistics \citep{HR2009}.

The distributed method is built upon the iterative,  multi-round algorithm proposed by  \cite{WKSZ2017} and \cite{JLY2018}, which only communicates gradient information at each round and therefore is communication-efficient.
By a delicate choice of local and global robustifications parameters,  the proposed estimator  satisfies exponential-type deviation bounds when the errors only have finite variance. Specifically,  we show that the distributed estimator, obtained by a few rounds of communications,   achieves the optimal centralized deviation bound  as if the data were pooled together and subject to sub-Gaussian errors. 
The robustification parameters are also self-tuned, making the algorithm computationally convenient. 
We further derive a Berry-Esseen bound for the distributed estimator, based on which we construct  robust confidence intervals.  Finally,  we propose a distributed penalized adaptive Huber regression estimator for high-dimensional sparse models, and establish its (near-)optimal theoretical guarantees.

\medskip
\noindent
{\sc Notation}: For each integer $k\geq 1$, we use $\RR^k$ to denote the the $k$-dimensional Euclidean space. The inner product of two vectors $u=(u_1, \ldots, u_k)^\T, v=(v_1, \ldots ,v_k)^\T \in \RR^k$ is defined by $u^\T v = \langle u, v \rangle= \sum_{i=1}^k u_i v_i$.
We use $\| \cdot \|_p$ $(1\leq p \leq \infty)$ to denote the $\ell_p$-norm in $\RR^k$: $\| u \|_p = ( \sum_{i=1}^k | u_i |^p )^{1/p}$ and $\| u \|_\infty = \max_{1\leq i\leq k} |u_i|$. 
For any $k\times k$ symmetric matrix $A \in \RR^{k\times k}$, $\| A \|_2$ is the operator norm of $A$. For a positive semidefinite matrix $A \in \RR^{k\times k}$,  $\| \cdot \|_{A}$ denotes the norm induced by $A$, that is, $\| u \|_{A} = \| A^{1/2} u \|_2$, $u \in \RR^k$.  
Moreover,  we use $\mathbb{S}^{k-1}= \{ u \in \RR^k: \| u \|_2 = 1\}$ to denote the unit sphere in $\RR^k$.
For two sequences of non-negative numbers $\{ a_n \}_{n\geq 1}$ and $\{ b_n \}_{n\geq 1}$, $a_n \lesssim b_n$ indicates that there exists a constant $C>0$ independent of $n$ such that $a_n \leq Cb_n$; $a_n \gtrsim b_n$ is equivalent to $b_n \lesssim a_n$; $a_n \asymp b_n$ is equivalent to $a_n \lesssim b_n$ and $b_n \lesssim a_n$.

\section{Distributed Adaptive Huber Regression}
\label{sec2}

\subsection{Distributed Huber regression with adaptive robustification parameters}
\label{sec2.1}

Consider a linear regression model
\#
	y_i =   x_i^\T \beta^* + \varepsilon_i  , \quad  \EE( \varepsilon_i | x_i ) = 0 , \ \ i=1,\ldots, N, \label{linear.model}
\#
where $x_i = (x_{i1}, \ldots, x_{ip})^\T$ with $x_{i1} \equiv 1$ is the covariate for the $i$th individual, and $\beta^* \in \RR^p$  is the underlying coefficient vector. 
This setting allows conditional heteroscedastic models,  where $\varepsilon_i$ can depend on $x_i$. For example,  in a local-scale model we have $\varepsilon_i = \sigma(x_i) e_i$, where $\sigma(x_i)$ is a function of $x_i$, and $e_i$ is independent of $x_i$. 
In the absence of normality assumption on the (conditional) error distribution, Huber's $M$-estimator \citep{H1973} is one of the most widely used robust alternative to the least squares estimator.
Given some $\tau>0$,  referred to as the {robustification parameter},  Huber's regression $M$-estimator  for estimating $\beta^*$ is defined as
\$
\wh \beta  = \wh \beta_\tau \in \argmin_{\beta \in \RR^p}     ~\wh \cL_\tau(\beta)  := \frac{1}{N} \sum_{i=1}^N \ell_\tau(y_i -   x_i^\T  \beta  ),
\$
where $	\ell_\tau(u) = 0.5 u^2 I(|u| \leq \tau) +  (\tau |u| - 0.5 \tau^2) I(|u|>\tau) $ is the Huber loss.  Traditionally, $\tau$ is often chosen to be $1.345 \sigma$ with $\sigma$ either determined by a robust scale estimate or simultaneously estimated by solving a system of equations,
in order to achieve 95\% asymptotic relative efficiency  while gaining robustness when there are contaminated or heavy-tailed symmetric errors \citep{B1975,W1995}.  In the presence of asymmetric heavy-tailed errors, \cite{FLW2017} and \cite{SZF2020} proposed (regularized) adaptive Huber regression estimators with $\tau$ scaling with the sample size and parametric dimension, and established exponential-type deviation bounds when $\varepsilon_i$'s only have finite $(1+\delta)$-th moments for some $0<\delta\leq 1$.

In the linear model \eqref{linear.model},  we allow heteroscedastic errors that are of the form $\varepsilon_i = \sigma(x_i) e_i$, where $\sigma(\cdot)$ is an unknown function on $\RR^p$ and $e_i$ is independent of $x_i$.   When the error variables $\varepsilon_i$ are heavy-tailed,  asymmetric and have finite variance $\sigma^2$,  \cite{SZF2020} showed that Huber's estimator $\wh \beta_\tau$ with $\tau \asymp  \sigma \sqrt{N/ (p+\log N)}$, referred to as the adaptive Huber regression (AHR) estimator,  exhibits sharp finite-sample deviation properties \citep{C2012}, while the least squares estimator is far less concentrated around $\beta^*$. We say $\varepsilon_i$ is heavy-tailed if it has infinite $k$-th absolute moment for some $k>2$.   

In  the distributed setting,  assume that the overall dataset $\{ (y_i, x_i) \}_{i=1}^N$ is stored on $m$ node machines,  one central machine and $m-1$ local machines that connected to the central. For $j =1,\ldots, m$, the $j$th machine stores a subsample of $n_j$ observations, denoted by $\{ (y_i, x_i) \}_{i\in \cI_j}$, and $\cI_j$'s are disjoint index sets that satisfy $\cup_{j=1}^m \cI_j = \{ 1,\ldots, N\}$ and $N= \sum_{j=1}^m | \cI_j |  =\sum_{j=1}^m n_j$.  Without loss of generality, we assume $n_1 = \cdots = n_j = n $ and $N= n \cdot m$ is divisible by $m$. We thus refer to $n$ as the local sample size.
When the entire dataset is available,  the optimal $\tau$  scales with the total sample size $N$ and dimension $p$  for optimal bias and robustness tradeoff. With decentralized data,  each local machine only has access to a subsample, so that the  ``locally optimal" $\tau$ depends on the local sample size. This, however, will lead to sub-optimal bounds for the aggregated estimator because $\tau$ is not large enough to offset the bias. To parallelize AHR in a distributed setting without compromising statistical optimality,  we introduce two robustification parameters $\tau$ and $\kappa$, referred to as the global and local robustifiation  parameters, and define the global and local Huber loss functions as $	 \wh \cL_\tau (\beta)   = (1/N) \sum_{i=1}^N \ell_\tau(y_i - x_i^\T \beta )$ and $\wh \cL_{j, \kappa } (\beta) = (1/n) \sum_{i \in \cI_j}  \ell_\kappa( y_i - x_i^\T \beta )$ for $j = 1,\ldots , m .$
Using this adaptive robustification procedure, we then extend the approximate Newton-type method \citep{SSZ2014, JLY2018} to robust regression with heavy-tailed skewed errors.  

Starting with an initial estimator $\wt \beta^{(0)}$ of $\beta^*$, we define the shifted  adaptive Huber loss
\#
	\wt \cL(\beta)  & =  \wh  \cL_{1,\kappa} (\beta)  - \big\langle \nabla  \wh  \cL_{1,\kappa } ( \wt \beta^{(0)} ) - \nabla  \wh \cL_\tau( \wt \beta^{(0)}),  \beta    \big\rangle   \nn \\
	& =  \wh  \cL_{1,\kappa} (\beta)  - \Big\langle \nabla  \wh  \cL_{1,\kappa } ( \wt \beta^{(0)} ) -  \frac{1}{m} \sum_{j=1}^m \wh \cL_{j, \tau } ( \wt \beta^{(0)}),  \beta   \Big\rangle  ,   \ \ \beta \in \RR^p.  \label{surrogate.loss}
\#
Implicitly the shifted loss $\wt \cL(\cdot)$ depends on both local and global robustification parameters $\kappa$ and $\tau$.
It uses data available only on the first machine, used as the central machine, along with $p$-dimensional gradient vectors $\wh \cL_{j,\kappa}(\wt \beta^{(0)})$ ($j=2,\ldots, m)$ that were sent from the remaining local machines.  
The ensuing one-step estimator is given by
\#
	\wt \beta^{(1)} = \wt \beta^{(1)} _{\kappa ,\tau} \in \argmin_{\beta \in \RR^p}  	\wt \cL(\beta) .  \label{one.step.huber}
\#
This procedure requires one communication round of $O(pm)$ bits, and thus is communication-efficient. To investigate the statistical properties of $\wt \beta^{(1)}$, we impose the following moment  condition on the data generating process.
 
\noindent
(C1).
The predictor $x \in \RR^p$ is {sub-Gaussian}: there exists $ \upsilon_1\geq (2\log 2)^{-1/2}$ such that $\PP( | z^\T u |    \geq \upsilon_1    t  ) \leq  2e^{-t ^2/2}$ for every unit vector $u \in \mathbb{S}^{p-1}$ and $t \geq 0$, where $z = \Sigma^{-1/2} x$ and $\Sigma = \EE(x x^\intercal)$ is positive definite. Moreover, the regression error $\varepsilon$ satisfies $\EE(\varepsilon |x) = 0$ and $\EE(\varepsilon^2 | x ) \leq \sigma^2$ almost surely.

For prespecified parameters $r, r_*>0$, define the events
\#
\cE_0 (r) = \big\{ \wt \beta^{(0)} \in \Theta(r) \big\} ~~\mbox{ and }~~  \cE_*(r_*)  =\big\{   \|   \nabla \wh  \cL_\tau(\beta^* )  \|_{\Omega} \leq r_* \big\}  , \label{event0}
\#
where $\Theta(r)  :=   \{  \beta\in \RR^p: \| \beta - \beta^* \|_{\Sigma} \leq r   \}$ and $\Omega := \Sigma^{-1}$.
Here $r$ quantifies the statistical accuracy of the initial estimator $\wt \beta^{(0)} $,  and $r^*$ determines the estimation error  of the centralized AHR estimator which essentially depends on the score $ \nabla \wh  \cL_\tau(\beta^* )$ with the global robustification parameter.
 
\begin{theorem} \label{thm:one-step}
Assume Condition~(C1) holds.  For any $u>0$, let the robustification parameters satisfy  $\tau \geq \kappa \asymp \sigma\sqrt{n/(p+u)}$,  and suppose the local sample size satisfies $n\gtrsim p+u$.
Then, conditioned on the event $\cE_0(r_0) \cap  \cE_*(r_*)$ with $8 r_* \leq r_0  \leq  \sigma$, the one-step estimator $\wt \beta^{(1)}$ defined in \eqref{one.step.huber} satisfies 
\begin{gather}
	 \| \wt \beta^{(1)} - \beta^*  \|_{\Sigma} \lesssim  \sqrt{\frac{p+u}{n}} \cdot r_0     +   r_* \label{one-step.error} \quad\text{and}\\
    \|   \wt \beta^{(1)} - \beta^*  +\Sigma^{-1}\nabla \wh \cL_\tau(\beta^* )   \|_{\Sigma} \lesssim  \sqrt{\frac{p+u}{n}} \cdot    r_0   \label{one-step.bahadur}
\end{gather}
with probability at least $1-3 e^{-u}$.
\end{theorem}

In the above theorem, the bound \eqref{one-step.error} reflects the delicate dependence of the one-step error on the initial error $r_0$ as well as the centralized error rate $r_*$.  If we take $\wt \beta^{(0)}$ to be a local estimator constructed on a single local machine that has access to only $n$ observations, we may expect a sub-optimal convergence rate $r_0 \asymp \sigma  \sqrt{p/n}$.  Moreover,  it can be shown that $\| \nabla \wh  \cL_\tau(\beta^* )  \|_{\Omega} \lesssim  \sigma\sqrt{p/N} + \sigma^2/\tau  + \tau p/N$ with high probability,  up to  logarithmic factors; see Lemma~\ref{lem:global.score} in the Supplementary Material.  Hence,  the choice of $r_*$ corresponds to the optimal  rate of convergence when the entire dataset is available and $\tau \asymp \sigma \sqrt{N/p}$.  
Under the prescribed sample size scaling $n\gtrsim p $,  the one-step estimator $\wt \beta^{(1)}$ refines the statistical accuracy of $\wt \beta^{(0)}$ by a factor of order $\sqrt{p/n}$, which is strictly less than 1.  We thus expect the multi-step estimator,  with sufficiently many communication rounds, will achieve the optimal convergence rate obtainable on the entire dataset.

The proposed multi-round procedure for adaptive Huber regression is iterative,  starting at iteration 0 with an initial estimate $\wt \beta^{(0)} \in \RR^p$. At iteration $t\geq 1$,  it updates the estimate $\wt \beta^{(t)}  $ by fitting a shifted adaptive Huber regression which leverages global first-order information, depending on $\tau$, and local higher-order information, depending on $\kappa$. 
The procedure involves two steps.

\noindent
{\sc 1. Communicating gradient information}.  The central machine broadcasts $\wt \beta^{(t-1)}$ to every local machine. The $j$th machine, $1\leq j \leq m$,  computes the gradient $\nabla \cL_{j,\tau} (\wt \beta^{(t-1)})$, and sends it back to the central machine.  This step requires a communication of $2(m-1)p$ bits.

\noindent
{\sc 2. Fitting local shifted AHR}.  The central machine computes the update $\wt \beta^{(t)}$, defined as a solution to the optimization problem
\#
\min_{\beta \in \RR^p}  	~ \wt  \cL^{(t)} (\beta )  := \wh  \cL_{1,\kappa } ( \beta )  - \bigg\langle \nabla \wh  \cL_{1, \kappa } ( \wt \beta^{(t-1)} ) - \frac{1}{m} \sum_{j=1}^m \nabla \wh \cL_{j, \tau}(\wt \beta^{(t-1)} ) , \beta \bigg\rangle , \label{surrogate.loss.l}
\#
which can be solved by the method of iteratively reweighted least squares or quasi-Newton methods.  Details are given in section \ref{gdbb}.
We summarize the procedure, with an early stopping criterion, in Algorithm~\ref{algo:dqr}.

\begin{algorithm}[!t]
    \caption{ {\small  Communication-Efficient Adaptive Huber Regression.}}
    \label{algo:dqr}
    Input: data batches $\{(y_i, x_i)\}_{i\in \cI_j}$, $j=1,\ldots, m$, stored on $m$  machines, robustification parameters $\tau \geq \kappa >0$, initialization $\wt{\beta}^{(0)}$, number of iterations $T, g_0 = 1$.
    \begin{algorithmic}[1]
      \FOR{$t = 1, 2 \ldots, T$}
       \STATE  Broadcast $\wt \beta^{(t-1)}$ to all local machines;
          \STATE The $j$th ($1 \leq j\leq m$) machine computes $\nabla \hat \cL_{j,\tau}(\wt \beta^{(t-1)})$, and transmit it to the central machine;
          \STATE Compute $\nabla \hat \cL_\tau(\wt \beta^{(t-1)}) = (1/m) \sum_{j=1}^m  \nabla \hat \cL_{j,\tau}(\wt \beta^{(t-1)})$, $
       \nabla \hat \cL_{1, \kappa }(\wt \beta^{(t-1)})$ and $g_t = \| \nabla \hat \cL_\tau(\wt \beta^{(t-1)})  \|_{\infty}$  on the central machine;
          \STATE If $g_t \geq g_{t-1}$ or $g_t \leq 10^{-5}$ break ; otherwise,  on the central machine,  solve the shifted adaptive Huber regression problem in \eqref{surrogate.loss.l} to update the estimate   $\wt \beta^{(t)}$;
      \ENDFOR 
    \end{algorithmic}
      Output: $\wt \beta^{(T)}$.
\end{algorithm}

\begin{theorem} \label{thm:multi-step}
Assume the same conditions in Theorem~\ref{thm:one-step}, and let $8 r_* \leq r_0 \leq \sigma$.  Conditioned on event $\cE_0(r_0) \cap \cE_*(r_*)$, the distributed AHR estimator $\wt \beta^{(T)}$ with $T \gtrsim \lceil  \log(r_0/r_*) /  \log(n/(p+u))    \rceil$ satisfies the bounds
\#
 \| \wt \beta^{(T)} - \beta^*   \|_{\Sigma} \lesssim r_* ~~\mbox{ and }~~
 \|      \wt \beta^{(T)}- \beta^* + \Sigma^{-1} \nabla \wh \cL_\tau(\beta^*)  \|_{\Sigma}  \lesssim     \sqrt{\frac{p+u}{n}}  \cdot r_* \label{multi-step.bound}
\#
with probability at least $1- (2T+1) e^{-u}$.
\end{theorem}

The above result shows that, with proper choices of $\tau$ and $\kappa$ as well as the number of iterations, the statistical error of the multi-step distributed AHR estimator matches that of the centralized AHR estimator on the entire dataset.For the initialization, we may take $\wt \beta^{(0)}$ to be a local AHR estimator computed on the central machine. 
With the above preparations,  we are ready to explicitly describe the estimation error and Bahadur linearization error of the proposed distributed AHR estimator.  The result is nonasymptotic, and carefully tracks the impact of the parametric dimension $p$, local sample size $n$ and the number of machines $m$.

\begin{theorem} \label{thm:final.rate}
Assume Condition~(C1) holds, and suppose the local sample size satisfies $n\gtrsim p+\log n + \log_2 m$, where $\log_2 m := \log(\log m)$ and $m= N/n$. Choose the robustification parameters $\tau \geq \kappa >0$ as $\tau \asymp \sigma \sqrt{N/(p+\log n + \log_2 m)}$ and $\kappa \asymp \sigma\sqrt{n/(p+\log n+\log_2 m)}$. Then, starting at iteration 0 with a local AHR estimate $\wt \beta^{(0)}$, the distributed estimator $\wt \beta = \wt \beta^{(T)}$ with $ T \asymp \lceil  \frac{\log(m)}{\log( n / (p+\log n+ \log_2 m))}  \rceil$   satisfies 
\begin{gather}
 \| \wt \beta  - \beta^*   \|_{\Sigma} \lesssim \sigma \sqrt{\frac{p+ \log n + \log_2 m}{N}}\quad\text{and} \label{final.rate}\\
 \left\| \wt \beta  - \beta^* - \Sigma^{-1} \frac{1}{N}\sum_{i=1}^N  \psi_\tau(\varepsilon_i) x_i  \right\|_{\Sigma} \lesssim   \sigma   \frac{p+\log n + \log_2 m}{(n N)^{1/2}}  \label{final.bahadur}
\end{gather}
with probability at least $1- C n^{-1}$, where $\psi_\tau (u) := \ell_\tau'(u) = \sign(u) \min(|u|, \tau)$.
\end{theorem}

The above theorem indicates that  the multi-step distributed AHR estimator $\wt \beta$ achieves the optimal statistical rate of convergence by a delicate
combination of the local robustification parameter,  the global robustification parameter,  and number of communication rounds.  The second bound, \eqref{final.bahadur},  explicitly describes the
error term of the Bahadur linearization. This  allows to establish the asymptotic distribution of $\wt \beta$ when both $p,  n$ tend to infinity. Moreover,  to achieve statistical optimality and communication efficiently simultaneously,  the above results impose minimal conditions on the number of machines $m$.   In summary,  when data are  heavy-tailed and collected on each machine remain local,  the proposed procedure delivers a statistically optimal estimate by communicating as many as $O( p m \log (m)  )$ bits.

\subsection{Distributed confidence estimation}\label{dist.conf.construction}

In this section, we consider uncertainty quantification of the multi-step estimator in a distributed setting, with a particular focus on statistical confidence estimation.
We first establish a  Berry-Esseen bound for linear functionals of the distributed AHR estimator $\wt \beta$,  which  explicitly quantifies the normal approximation error.

\begin{theorem} \label{thm:CLT}
In addition to the conditions in Theorem~\ref{thm:final.rate}, assume $\EE (\varepsilon^2 | x) = \sigma^2$ and $\EE (|\varepsilon|^{2+\delta} | x) \leq v_{2+\delta}$ almost surely for some $0<\delta\leq 1$. Then, the distributed estimator $\wt \beta = \wt \beta^{(T)}$ satisfies
\#
\sup_{t\in  \RR, \, a \in \RR^p } &  \left|  \PP \left[       \frac{ N^{1/2} a^\T (\wt \beta - \beta^* ) }{  \sqrt{ \EE \{ \psi_\tau(\varepsilon)  a^\T \Sigma^{-1}x  \}^2 } }  \leq   t \right]   - \Phi(t)  \right| \nn \\
&  \lesssim \frac{p+\log n+\log_2 m}{n^{1/2}} +  \frac{v_{2+\delta }(p+ \log n + \log_2 m)^{(1+\delta)/2}}{ \sigma^{2+\delta} N^{\delta/2}},  \label{CLT1}
\#
where $\Phi(\cdot)$ is the standard normal distribution function. In particular, assume $\EE(|\varepsilon|^3|x) \leq v_3<\infty$ almost surely. Then, under the dimension constraint $p + \log_2 m = o(n^{1/2})$, 
\#
  \frac{ N^{1/2} a^\T (\wt \beta - \beta^* ) }{ \sqrt{ \EE \{ \psi_\tau(\varepsilon)  a^\T \Sigma^{-1}x  \}^2 }  } \xrightarrow {\rm d} \cN(0, 1)~~\mbox{ and }~~
\frac{N^{1/2} a^\T (\wt \beta - \beta^* )}{ \sigma ( a^\T \Sigma^{-1} a )^{1/2} } \xrightarrow {\rm d} \cN(0, 1)   \label{CLT2}
\#
uniformly over $a \in \RR^p$ as $n \to \infty$.
\end{theorem}

Let $\wt \beta = (\wt \beta_1, \ldots, \wt \beta_p)^\T$ be the distributed estimator described in the previous subsection.  Theorem~\ref{thm:CLT} implies that, for each $1\leq j\leq p$,  $N^{1/2}   (\wt \beta_j - \beta_j^* )$ is asymptotically normal with zero mean and variance  $(\Sigma^{-1} \EE \{ \psi_\tau(\varepsilon) x x^\T  \}^2 \Sigma^{-1})_{jj}$.  Let $\wh \Sigma = (1/N) \sum_{i=1}^N x_i x_i^\T$ be the sample version of $\Sigma$, and   $\wh \varepsilon_i  = y_i - x_i^\T \wt \beta$ be the fitted residuals. 
It can be shown that  
$
	 (  \wh \Sigma^{-1}  N^{-1}\sum_{i=1}^N \psi_\tau^2(\varepsilon_i) x_i x_i^\T \, \wh \Sigma^{-1} )_{jj}
$
provides a consistent estimator of $(\Sigma^{-1} \EE \{ \psi_\tau(\varepsilon) x x^\T  \}^2 \Sigma^{-1})_{jj}$.  In a distributed setting,  the computation of this variance estimator requires communicating $O(p^2 m)$ bits, thus incurring exorbitant communication costs when $p$ is large.

To achieve a tradeoff between communication and statistical efficiencies,  we propose averaging pointwise variance estimators, defined by $\wh \sigma_j^2 : = (1/m) \sum_{k=1}^m \wh \sigma_{jm}^2$ for $j=1,\ldots, p$, where 
$$
   \wh \sigma^2_{jk} = ( \wh \Sigma_k^{-1} \wh \Lambda_k \wh \Sigma_k^{-1} )_{jj}, \quad  \wh \Lambda_k = \frac{1}{n} \sum_{i \in \cI_k}  \psi^2_\tau(\wh \varepsilon_i) x_i x_i^\T  ~\mbox{ and }~    \wh \Sigma_k = \frac{1}{n} \sum_{i\in \cI_k} x_i x_i^\T.
$$
This approach takes one additional round of communication, and is robust against heteroscedastic errors that are of the form $\varepsilon_i = \sigma(x_i) e_i$. When $\varepsilon_i$ is independent of $x_i$,  the asymptotic variances reduce to $\EE \{ \psi^2_\tau(\varepsilon) \} (\Sigma^{-1})_{jj}$,  and thus can be consistently estimated by $\wt \sigma^2_j  :=   (\wh \sigma^2_\varepsilon /m)\sum_{k=1}^m(\wh \Sigma_k^{-1})_{jj}$, where $\wh \sigma^2_\varepsilon = (N-p)^{-1} \sum_{i=1}^N \psi^2_\tau(\wh \varepsilon_i )$.
For $\alpha\in(0,1)$,  the  distributed 100$(1-\alpha)$\% normal-based confidence intervals for $\beta^*_j$, $j=1,\ldots,p$, are given by $[ \wt \beta_j -z_{\alpha/2}\wh \sigma_j N^{-1/2} ,    \wt \beta_j + z_{\alpha/2}\wh \sigma_j N^{-1/2}]$ or $[ \wt \beta_j - z_{\alpha/2}\wt \sigma_j N^{-1/2} ,    \wt \beta_j + z_{\alpha/2}\wt \sigma_j N^{-1/2}]$, where $z_{\alpha/2} =   \Phi^{-1}(1-\alpha/2)$.

\section{Distributed Regularized Adaptive Huber Regression}

In this section, we consider high-dimensional linear models under sparsity.  Specifically,  we allow the parametric dimension $p$ to be much larger than the local sample size $n$, and assume $\beta^*$ is $s$-sparse, where $s = |\cS|$ and   $\cS =  {\rm supp}(\beta^*) = \{ 1\leq j\leq p: \beta^*_j \neq 0\}$ denotes the true active set.

Given independent observations $\{ (y_i, x_i ) \}_{i=1}^N$ from the linear model \eqref{linear.model}, the centralized/global $\ell_1$-penalized Huber regression estimator ($\ell_1$-Huber) is defined as
\#
\wh \beta = \wh \beta_{\tau }(\lambda)    \in \argmin_{\beta \in \RR^p}   \big\{  \wh \cL_\tau(\beta) + \lambda \| \beta \|_1  \big\}  , \label{lasso-Huber}
\#
where $\lambda>0$ is a regularization parameter.  Statistical properties of $\ell_1$-penalized Huber regression have been thoroughly studied by \cite{LZ2011}, \cite{FLW2017},  \cite{L2017} and \cite{CLL2020} from different perspectives.   To deal with asymmetric heavy-tailed errors, \cite{FLW2017} established high probability bounds for the $\ell_1$-Huber estimator with $\tau \asymp \sigma \sqrt{N/\log(p)}$ in the high-dimensional regime $p\gg  n\gtrsim s \log(p)$.

\begin{remark}
In practice, it is natural to leave the intercept or a given subset of the parameters unpenalized in the penalized $M$-estimation framework \eqref{lasso-Huber}. Denote by $\cR \subseteq \{1,\ldots, p\}$ the index set of unpenalized parameters, which is typically user-specified and contains at least index 1.  A more flexible $\ell_1$-Huber estimator can be obtained by solving
$
\min_{\beta \in \RR^p}   \{  \wh \cL_\tau(\beta) + \lambda \| \beta_{\cR^{{\rm c}}} \|_1      \}  = \min_{\beta \in \RR^p}   \{  \wh \cL_\tau(\beta) + \lambda \sum_{j\in \mathcal{R}^{{\rm c}}} |\beta_j|     \} .
$
Similar theoretical analysis can be carried out with slight modifications, and thus will be omitted. 
\end{remark}

In a distributed setting,  we integrate the ideas of \cite{WKSZ2017} and \cite{JLY2018} with adaptive robustification to parallelize regularized Huber regression with controlled communication cost and optimal statistical guarantees.
As before, let $\tau$ and $\kappa$ be the global and local robustification parameters. 
Recall that  $\wh \cL_{j, \kappa}(\cdot)$, $j=1,\ldots, m$, denote   local Huber loss functions.  Commenced with a regularized estimator $\wt \beta^{(0)}$, let $\wt \cL(\beta) =  \wh  \cL_{1,\kappa} (\beta)  -   \langle \nabla  \wh  \cL_{1,\kappa } ( \wt \beta^{(0)} ) - \nabla  \wh \cL_\tau( \wt \beta^{(0)}),  \beta  \rangle$ be the shifted adaptive Huber loss as in \eqref{surrogate.loss}.  With slight abuse of notation, we define the one-step $\ell_1$-penalized Huber regression estimator as
\#
  \wt \beta^{(1)} = \wt \beta^{(1)}_{\kappa, \tau}(\lambda)  \in \argmin_{\beta \in \RR^p}   \big\{   \wt \cL(\beta) + \lambda \| \beta \|_1 \big\} . \label{lasso.shifted.huber}
\#
To assess the statistical properties of the one-step estimator $\wt \beta^{(1)}$,  we work under the the following moment condition on the random covariate vector in high dimensions.
 \medskip
 
\noindent
(C2).
The covariate vector $x = (x_1, \ldots, x_p)^\T \in \RR^p$ with $x_1 \equiv 1$ has bounded components and uniformly bounded kurtosis. That is, $\max_{1\leq j\leq p} |x_j| \leq B$ for some $B\geq 1$ and $\mu_4 =\sup_{u \in \mathbb S^{p-1}} \EE (z^\T u)^4<\infty$, where $z = \Sigma^{-1/2}x $ and $\Sigma= (\sigma_{jk})_{1\leq j,k\leq p} = \EE(x x^\T)$. Write $\sigma_u = \max_{1\leq j\leq p} \sigma_{jj}^{1/2}$ and $\lambda_l = \lambda_{\min}(\Sigma) >0$.  For simplicity,  we assume $\lambda_l=1$.  Moreover, the error variables $\varepsilon_i$ satisfy $\EE(\varepsilon_i |x_i) = 0$ and $\EE(\varepsilon_i^2 | x_i ) \leq \sigma^2$ almost surely. 

As before, we first examine the performance of $\wt \beta^{(1)}$ conditioned on certain ``good" events in regard of the initialization and the centralized $\ell_1$-Huber estimator.  For $r_0, \lambda_*>0$, define 
\#
\cE_0(r_0) = \big\{ \wt \beta^{(0)} \in \Theta(r_0) \cap \Lambda \big\} ~~\mbox{ and }~~
\cE_*(\lambda_*) = \big\{  \| \nabla \wh \cL_\tau(\beta^*) -  \nabla \cL_\tau(\beta^*) \|_\infty \leq \lambda_* \big\} , \label{hd.good.events}
\#
where $\Lambda := \{ \beta \in \RR^p: \| \beta -\beta^* \|_1 \leq 4s^{1/2} \| \beta - \beta^* \|_{\Sigma} \}$ is an $\ell_1$-cone.

\begin{theorem} \label{thm:hd.one-step}
Assume Condition~(C2) holds. Given $\delta \in (0,1)$ and $0< r_0, \lambda_* \lesssim \sigma$, let $(\tau, \kappa, \lambda )$ satisfy $\tau \geq \kappa \asymp \sigma\sqrt{n/\log(p/\delta)}$ and $\lambda = 2.5 (\lambda_* + \rho)$ with 
\$
	\rho \asymp  \max \Bigg\{   r_0    \sqrt{\frac{s\log(p/\delta)}{n}} ,   s^{-1/2}  \sigma^2 \tau^{-1}     \Bigg\} . 
\$
Moreover, suppose the local sample size satisfies $n\gtrsim s\log(p/\delta)$. Then, conditioned on the event $\cE_0(r_0) \cap \cE_*(\lambda_*)$, the one-step regularized estimator $\wt \beta^{(1)}$ defined in \eqref{lasso.shifted.huber} satisfies $\wt \beta^{(1)} \in \Lambda$ and 
\#
	\| \wt \beta^{(1)} - \beta^* \|_{\Sigma} \lesssim       s\sqrt{\frac{ \log(p/\delta)}{n}}   \cdot r_0 +  \sigma^2 \tau^{-1}   +    s^{1/2} \lambda_*  \label{one-step.hd.err}
\#
with probability at least $1-\delta$.
\end{theorem}

Theorem~\ref{thm:hd.one-step} indicates that the one-step procedure is able to reduce the statistical error of the initial estimator by a factor of $s\sqrt{\log(p)/n}$ when  the local sample size satisfies $n\gtrsim s^2 \log (p)$; see the first term on the right-hand of \eqref{one-step.hd.err}. The second term, $ \sigma^2 \tau^{-1}   +    s^{1/2} \lambda_* $, corresponds to the global error rate achievable on the entire dataset. In view of Theorem~B.2 (with $\delta=1$) in \cite{SZF2020}, if we take $\lambda_* \asymp \sigma \sqrt{\log(p)/N}$ and $\tau \asymp \sigma \sqrt{N/\log(p)}$, the centralized $\ell_1$-Huber estimator given in \eqref{lasso-Huber} satisfies 
$ \| \wh \beta - \beta^* \|_{\Sigma} \lesssim  \sigma^2 \tau^{-1}   +    s^{1/2} \lambda_* \asymp   \sigma \sqrt{s \log(p)/N} $ with probability at least $1-Cp^{-1}$.

Now we extend the iterative procedure in Section~\ref{sec2} to high-dimensional settings, starting at iteration 0 with an initial estimate $\wt \beta^{(0)}\in \RR^p$. At iteration $t=1,2,\ldots$,  it proceeds as follows:

\noindent
{\it Communicating gradient information}. The $j$th $(2\leq j\leq m)$ machine receives  $\wt \beta^{(t-1)}$ from the central machine,  computes the local gradient $\nabla \wh \cL_{j,  \tau }( \wt \beta^{(t-1)})$,  and sends it back to the central.

\noindent
{\it Fitting local regularized AHR}: On the central machine,  solve $\min_{\beta \in \RR^p} \{ \wt \cL^{(t)}(\beta) + \lambda_t \| \beta\|_1 \}$ to obtain  $\wt \beta^{(t)}$, where $ \wt \cL^{(t)}(\beta) =  \wh \cL_{1, \kappa }(\beta) - \langle  \nabla \wh \cL_{1, \kappa } (\wt \beta^{(t-1)} ) - (1/m) \sum_{j=1}^m \nabla \wh \cL_{j,\tau}(\wt \beta^{(t-1)} ), \beta \rangle$ and $\lambda_t>0$ is a regularization parameter.

Computationally,  we use a variant of the majorize-minimize algorithm \citep{LHY2000}, a proximal gradient descent type method, to solve the regularized optimization problem at each iteration. Details are provided in section \ref{lamm}.  Theorem \ref{thm:hd.final} below describes the statistical properties of the solution path $\{ \wt \beta^{(t)} \}_{t\geq 1}$ conditioned on a prespecified level of accuracy of the initial estimator.

\begin{theorem}\label{thm:hd.final}
Assume Condition~(C2) holds. Given $\delta \in (0,1)$ and $0< r_0, \lambda_* \lesssim \sigma$, let $(\tau, \kappa)$ satisfy $\tau \geq \kappa \asymp \sigma\sqrt{n/\log(p/\delta)}$. For $t=1,2,\ldots$, set $\lambda_t = 2.5(\lambda_* + \rho_t)>0$ with $\rho_t \asymp s^{-1/2} \max\{ \alpha^t r_0 , \sigma^2 \tau^{-1} \}$ and $\alpha \asymp s \sqrt{\log(p/\delta) /n}$. Suppose the local sample size satisfies $n\gtrsim s^2 \log(p/\delta)$, and let $r_* \asymp \sigma^2 \tau^{-1} + s^{1/2} \lambda_*$.   Then, conditioned on event $\cE_0(r_0) \cap \cE_*(\lambda_*)$, the distributed regularized estimator $\wt \beta^{(T)}$ with $T \asymp \frac{ \log(r_0/ r_* )}{\log(1/\alpha)}$ satisfies $\wt \beta^{(T)}  \in \Lambda$ and  $\| \wt \beta^{(T)} - \beta^* \|_{\Sigma} \lesssim r_*$ with probability at least $1-T \delta$.
\end{theorem}

With sufficiently many samples on the central machine---$n\gtrsim s^2 \log(p)$,  Theorems~\ref{thm:hd.one-step} and~\ref{thm:hd.final} ensure that the initial estimation error, albeit being sub-optimal,  can be repeatedly refined by a factor of order $s  \sqrt{\log(p)/n}$ until it reaches the optimal rate.  For simplicity,  we take $\wt\beta^{(0)}$ to be a local $\ell_1$-penalized AHR estimator, that is, $\wt\beta^{(0)} \in \argmin_{\beta \in \RR^p}  \{ \cL_{1,\kappa}(\beta) + \lambda_0 \| \beta \|_1 \}$.


\begin{corollary}\label{thm:hd.final2}
Assume Condition~(C2) holds, and the sample size per machine satisfies $n\gtrsim s^2 \log p$.  Choose the robustification and regularization parameters as $\tau \asymp \sigma \sqrt{N/\log(p)}$, $\kappa \asymp \sigma \sqrt{n/\log(p)}$ and 
\$
 \lambda_t \asymp  \sigma \sqrt{\frac{\log p}{N}} +   \sigma  \Bigg(  \frac{s^2 \log p}{n} \Bigg)^{t/2}  \sqrt{\frac{\log p }{n}}  , \ \ t =  0 , 1, 2, \ldots .
\$
Starting at iteration 0 with a local $\ell_1$-penalized AHR estimator,  the multi-step estimator $\wt \beta^{(T)}$ after $T \asymp  \lceil \log(m) \rceil $ rounds of communication satisfies the bounds
\#
 \| \wt \beta^{(T)} - \beta^* \|_{\Sigma} \lesssim \sigma \sqrt{\frac{s \log p}{N}} ~~\mbox{ and }~~   \| \wt \beta^{(T)} - \beta^* \|_1  \lesssim \sigma s \sqrt{\frac{  \log p}{N}}   \nn
\#
with probability at least $1-C \log(m)/p$.
\end{corollary}

Corollary~\ref{thm:hd.final2}, along with the global error analysis in \cite{FLW2017} and \cite{L2017},  implies the optimality of distributed adaptive Huber regression in terms of the tradeoff between communication cost and statistical accuracy.
 
 \begin{remark}
Under light-tailed error distributions (e.g., sub-Gaussian errors),  \cite{LLST2017} and \cite{BFLLZ2018} studied a one-shot approach based on averaging debiased Lasso estimators \citep{ZZ2014,vdG2014}. Theoretically,   averaged debiased Lasso achieves the optimal  error rate when the local size satisfies $n\gtrsim m s^2 \log(p)$; and computationally, each local machine needs to estimate a $p \times p$ matrix  for debiasing the Lasso.  
We may expect the same issues for the robust one-shot  method that averages debiased $\ell_1$-Huber estimators.
The proposed distributed AHR method not only requires the minimum sample complexity but also is computationally efficient.
 \end{remark}

\section{Optimization Methods}
\subsection{Barzilai-Borwein gradient descent for distributed AHR}\label{gdbb}

Let us first recall the multi-round distributed procedure for adaptive Huber regression.
Starting with an initial estimator $\wt \beta^{(0)} \in \RR^p$, and given robustification parameters $\tau$ and $\kappa$, for $t = 1,\ldots, T$, we update
\# 
\wt \beta^{(t)} \in \argmin_{\beta \in \RR^p}  \wt \cL^{(t)}(\beta) =  \wh  \cL_{1,\kappa} (\beta)  - \bigl\langle \nabla  \wh  \cL_{1,\kappa } ( \wt \beta^{(t-1)} ) - \nabla  \wh \cL_\tau( \wt \beta^{(t-1)}),  \beta    \bigr\rangle . \label{iterate.l}
\#
Since $\wt \cL^{(t)}(\cdot)$ is convex, twice-differentiable and provably locally strongly convex, we propose to use the gradient descent method with a Barzilai-Borwein update step (GD-BB) \citep{BB1988} to solve the optimization problem in \eqref{iterate.l}.
The Barzilai-Borwein method is motivated by quasi-Newton methods, which avoid calculating the inverse Hessian at each iteration. The latter is computationally expensive when $p$ is large.
To be specific, let us consider the optimization $\min_{\beta \in \RR^p}  \wt \cL^{(t)}(\beta)$ for a fixed $t\geq 1$. Starting with the initialization $\wt \bbeta^{ ( t,0 ) } = \wt \bbeta^{(t-1)}$,
at (inner) iteration $k = 1,2,...$, compute the update $\wt \bbeta^{ (t,k+1) }  =\wt \bbeta^{ (t,k) } - \min\{\eta_k,10\} \nabla \wt \cL^{(t)}(\wt \bbeta^{ (t,k) })$, where $\eta_1=1$ and for $k\geq 2$,
\#
\eta_k=\frac{\langle\wt {\beta}^{(t,k )}-\wt{\beta}^{(t,k-1)}, \wt{\beta}^{(t,k )}-\wt{\beta}^{(t,k-1)}\rangle}{\langle\wt{\beta}^{(t,k )}-\wt {\beta}^{(t,k-1)}, \nabla \wt {\cL}^{(t)}(\wt{\beta}^{(t,k )})-\nabla \wt{\cL}^{(t)}(\wt{\beta}^{(t,k-1)})\rangle}  \label{eta.update}
\#
or
\#
\eta_k=\frac{\langle\wt {\beta}^{(t,k )}-\wt{\beta}^{(t,k-1)},   \nabla \wt {\cL}^{(t)}(\wt{\beta}^{(t,k )})-\nabla \wt{\cL}^{(t)}(\wt{\beta}^{(t,k-1)}) \rangle}{  \|  \nabla \wt {\cL}^{(t)}(\wt{\beta}^{(t,k )})-\nabla \wt{\cL}^{(t)}(\wt{\beta}^{(t,k-1)}) \|_2^2 }  .\nn 
\#
In practice, the step size computed in GD-BB may sometimes vibrate to some extent, and this may cause instability of the algorithm. Therefore, we set a upper bound for the step sizes by taking $\min\{\eta_k, 10\}$. This procedure is summarized in Algorithm~\ref{algo:gdbb}.

\subsection{Majorize-minimize algorithm for distributed penalized AHR}\label{lamm}

In the high-dimensional setting, we need to solve $\ell_1$-penalized shifted Huber loss minimization problems.  
\begin{algorithm}[!t]
    \caption{ {\small  Gradient Descent with Barzilai-Borwein stepsize for solving \eqref{iterate.l}}}
    \label{algo:gdbb}
    Input: Local data vectors  $\{(y_{i}, x_{i})\}_{i \in \cI_{1}}$,  initial estimator  $\wh{\beta}^{0} = \wt {\beta}^{(t-1)}$,  gradient  $\nabla \wh{\cL}_{1, \kappa}(\wt {\beta}^{(t-1)}) $
and $ \nabla \wh{\cL}_{j,\tau}(\wt {\beta}^{(t-1)})$  for $ j=1, \ldots, m,$  and gradient tolerance level $\delta = 10^{-4}$.
    \begin{algorithmic}[1]
    \STATE Compute $ \wh{\beta}^{1} \leftarrow \wh{\beta}^{0}-\nabla \wt{\cL} ^{(t)} (\wh{\beta}^{0} ) $
      \FOR{$k = 1, 2 \ldots$}
       \STATE  Compute  $\eta_{k}$ as defined in  \eqref{eta.update}. 
       \STATE Update  $\wh{\beta}^{k+1} \leftarrow \wh{\beta}^{k}-\min\{\eta_{k},10\} \nabla \wt{ \cL}^{(t)} (\wh{\beta}^{k}  )$;
      \ENDFOR  ~when $ \|\nabla \wt{\cL}^{(t)} (\wh{\beta}^{k} ) \|_{\infty} \leq \delta $
    \end{algorithmic}
\end{algorithm}
With slight abuse of notation, given an initial regularized estimator $\wt {\beta}^{(0)}$, at each iteration $t = 1,2,\ldots,T$, define the update as
\#
   \wt \beta^{(t)} \in \argmin_{\beta \in \RR^p} \Bigg\{  \wt \cL^{(t)}(\beta) +  \lambda \| \beta_{-} \|_1 = \wh  \cL_{1,\kappa} (\beta)  - \big\langle \nabla  \wh  \cL_{1,\kappa } ( \wt \beta^{(t-1)} ) - \nabla  \wh \cL_\tau( \wt \beta^{(t-1)}),  \beta    \big\rangle + \lambda \| \beta_{-} \|_1 \Bigg\} . \label{hd.iterate.l}
\#
Here we use $\beta_- \in \RR^{p-1}$ to denote the subvector of $\beta$ with its first component removed. To solve the optimization problem in \eqref{hd.iterate.l}, we employ the locally adaptive majorize-minimize (LAMM) principle \cite{FLSZ2018}, which extends the classical MM algorithm \citep{HL2000} to accommodate $\ell_1$ penalty. This procedure minimizes a surrogate $\ell_1$-penalized isotropic quadratic function at each iteration, thus permitting an analytical solution.

Let $\wt \cL(\cdot)$ be the loss function of interest. For $k=1,2,...$, define 
\#
    g_k(\beta; \beta^{k-1}, \phi_k) = \wt{\cL}( \beta^{k-1})+\bigl\langle\nabla \wt{\cL}( {\beta}^{k-1}), \beta- {\beta}^{k-1}\bigr\rangle+\frac{\phi_{k}}{2}\|\beta-{\beta}^{k-1}\|_{2}^{2}. \nn
\#
We say $g_k(\beta;   \beta^{k-1},\phi_k)$ majorizes $\wt \cL (\beta)$ at $\beta^{k-1}$ if 
\#\label{local majorization}
g_k(\beta ;   \beta^{k-1}, \phi_k) \geq \wt \cL (\beta) ~~\forall \beta \in \RR^{p} \text{~~and~~} g_k(   \beta^{k-1} ;   \beta^{k-1}, \phi_k) = \wt \cL (  \beta^{k-1}).
\#

By choosing $\phi_k$ large enough, $g_k(\cdot ;  \beta^{k-1}, \phi_k)$ is guaranteed to satisfy \eqref{local majorization}. To find the smallest such $\phi_k$, we start with $\phi_0 = 0.0001$, and repeatedly inflate it by a constant factor, say 1.1, until \eqref{local majorization} is satisfied. Finally, we update $\beta^k$ by minimizing
\#\label{penalized LAMM loss}
    g_k( \beta ;   \beta^{k-1}, \phi_k)  + \lambda \| \beta_{-} \|_1.
\#
Due to the isotropic quadratic term in $g_k( \beta ;  \beta^{k-1}, \phi_k)$, $\beta^k$ can be obtained by a simple analytic formula:
\$
\left\{\begin{array}{l}
 {\beta}_{1}^{k}= {\beta}_{1}^{k-1}-\phi_{k}^{-1}  ( \nabla \wt{\cL}( {\beta}^{k-1}) )_1 \\
 {\beta}_{j}^{k}=S( {\beta}_{j}^{k-1}-\phi_{k}^{-1} ( \nabla  \wt{\cL}( {\beta}^{k-1}) )_j, \phi_{k}^{-1} \lambda) , ~~ j=2, \ldots, p,
\end{array}\right.
\$
where $S(u, \lambda) = \sign(u)\max(|u| - \lambda, 0)$ denotes the soft-thresholding operator. This algorithm also guarantees a descent in the overall loss function at every iteration, which is a direct consequence of \eqref{local majorization} and \eqref{penalized LAMM loss}:
\$
\begin{aligned}
\wt{\cL}( {\beta}^{k})+\lambda\| {\beta}_{-}^{k}\|_{1} & \leq g_k(  {\beta}^{k} ;    {\beta}^{k-1}, \phi_{k}) + \lambda\| {\beta}_{-}^{k}\|_{1} \\
& \leq g_k( {\beta}^{k-1} ; {\beta}^{k-1}, \phi_{k}) + \lambda\|{\beta}_{-}^{k-1}\|_{1}=\wt{\cL}({\beta}^{k-1})+\lambda\|  {\beta}_{-}^{k-1}\|_{1}.
\end{aligned}
\$
Algorithm~\ref{algo:lamm} summarizes the LAMM algorithm described above.

\begin{algorithm}[!t]
    \caption{ {\small  Local adaptive majorize-minimise (LAMM) algorithm for solving \eqref{lasso.shifted.huber}}}
    \label{algo:lamm}
    Input: Local data vectors  $\{(y_{i}, x_{i})\}_{i \in I_{1}}$, initial estimator $\wh{\beta}^{0}=\wt{\beta}^{(t-1)}$  gradient vectors  $\nabla \wh{\cL}_{1, \kappa}(\tilde{\beta}^{(t-1)})$  and  $\nabla \wh{\cL}_{\tau}(\tilde{\beta}^{(t-1)})$,  regularisation parameter  $\lambda$,  initial isotropic parameter  $\phi_{0}$ and convergence tolerance  $\delta$ 
    \begin{algorithmic}[1]
      \FOR{$k = 1, 2 \ldots$}
       \STATE  $\operatorname{Set} \phi_{k} \leftarrow \max \{\phi_{0}, \phi_{k-1} / 1.1\}$
       \REPEAT
        \STATE  Update  $\wh{\beta}_{1}^{k} \leftarrow \wh{\beta}_{1}^{k-1}-\phi_{k}^{-1} \nabla_{\beta_{1}} \wt{\cL}(\wh{\beta}^{k-1})$
        \STATE Update  $\wh{\beta}_{j}^{k} \leftarrow S(\wh{\beta}_{j}^{k-1}-\phi_{k}^{-1} \nabla_{\beta_{j}} \wt{\cL}(\wh{\beta}^{k-1}), \phi_{k}^{-1} \lambda) $ for  $j=2, \ldots, p $
        \STATE If $g_k(\wh{\beta}^{k} ; \wh{\beta}^{k-1},  \phi_{k})<\wt{\cL}(\wh{\beta}^{k})$,~set $ \phi_{k} \leftarrow 1.1 \phi_{k}$
       \UNTIL $g_k(\wh{\beta}^{k} ;  \wh{\beta}^{k-1}, \phi_{k}) \geq \wt{\cL}(\wh{\beta}^{k}) $
      \ENDFOR ~when $\|\wh{\beta}^{k} -\wh{\beta}^{k-1}\|_{2} \leq \delta$ 
    \end{algorithmic}
\end{algorithm}

\section{Numerical Studies}
In this section, we compare the numerical performance of the proposed method with several state-of-the-art distributed regression methods in both low and high dimensions.
\subsection{Distributed robust regression and inference}  \label{numerical.low.dimension}

In the low-dimensional setting where $n\gg p$, we consider five distributed regression methods: (i) the global adaptive Huber regression (AHR) estimator \citep{SZF2020} that uses all the available $N = mn$ observations; (ii) divide-and-conquer AHR (DC-AHR) estimator based on averaging $m$ local AHR estimators; (iii)  DC-OLS estimator that averages $m$ local OLS estimators;  (iv) distributed OLS estimator \citep{SSZ2014}; and (v) the proposed distributed AHR estimator with early stopping.

To implement methods (i) and (ii), we use the self-tuning principle proposed by \cite{WZZZ2020} which automatically selects the robustification parameter $\tau$. 
The distributed procedures (iv) and (v) are iterative, and require a reasonably well initial estimator, say $\wt \beta^{(0)}$. 
In our simulations, we take $\wt \beta^{(0)}$ to be either the DC-AHR or the DC-OLS estimator, which only requires one communication round.
When the error distribution is heavy-tailed and symmetric, DC-AHR often has better finite-sample performance than DC-OLS. However, it produces biased estimate  when the error is asymmetric. In contrast, although the DC-OLS exhibits larger variability due to heavy-tailedness, it has smaller bias on average. Therefore, we use DC-OLS estimator as the initialization for both methods (iv) and (v).
Recall that the distributed AHR estimator involves two robustification parameters $\kappa$ and $\tau$. The local parameter $\kappa$ can be automatically obtained by the self-tuning procedure \citep{WZZZ2020}. Guided by theoretical orders of $(\kappa, \tau)$ stated in Theorem~\ref{thm:final.rate}, we choose the global parameter $\tau$ to be $c m^{1/2}\kappa$, where $c\geq 1$ is a numerical constant that can be tuned by the validation set approach.
We suggest to choose $c$ from $\{1,2,3,4,5\}$,  which suffices to achieve promising performance in a wide range of simulation settings.

We generate data vectors $\{ (y_i , x_i) \}_{i=1}^N$ from a heteroscedastic model $y_i = x_i^\T\beta^* + c^{-1}(x_i^\T\beta^*)^2\varepsilon_i$,  where $\beta^* = (1.5,\ldots, 1.5)^\T \in \mathbb R^p, x_i = ( 1, x_{i2}, \ldots, x_{ip})^\T$ with $x_{ij} \sim \cN(0,1)$ for $j=2,\ldots, p$ and $c = \sqrt{3} \|\beta^*\|_2^2$ that makes $\EE \{c^{-1}(x_i^{\T}\beta^*)^2\}^2 = 1$. The regression errors $\varepsilon_i$ are generated from one of the following four distributions (centered if the mean is nonzero): (a) $\cN(0,1)$ (standard normal), (b) $t_2$ ($t$-distribution with 2 degrees of freedom), (c) Par$(4,2)$--Pareto distribution with scale  parameter 4 and shape parameter 2,  and (d) Burr$(1, 2, 1)$--Burr distribution or the Singh-Maddala distribution \citep{SM1976}, which is commonly used to model household income.  
First,  we fix $(n,p)= (400,20)$ and let the number of machines $m$ increase from 10 to 500.
Figure \ref{low_dim_estimation} plots the $\ell_2$-error $\|\wh \beta - \beta^*\|_2$ versus  the number of machines,  averaged over 500 replications, for all five methods.
The global and distributed AHR estimators have almost identical performance,  thus corroborating our theoretical results.
The DC-AHR estimator only performs well under symmetric errors and suffers from non-negligible bias if the errors come from asymmetric distributions.  This is largely expected because the robustification parameter for a local AHR estimator is tuned by a small subset of the data and results in a bias scaling with the local sample size. After averaging, this bias will not be offset when the number of machines increases. This points out a key drawback of the  one-shot averaging approach when dealing with skewed data distributed across local machines. The DC-OLS method has decaying estimation error as $m$ grows, but at a slower rate compared to the global and the distributed AHR estimators.  The boxplots in Figure~\ref{low_dim_box} show that the DC-OLS method often produces very poor estimates with high variability, while the distributed AHR method exhibits high degree of robustness.

Turning to uncertainty quantification, we construct approximate 95\%  confidence intervals for the slope coefficients based on distributed OLS and AHR methods.  As before,  we set $(n,p) = (400,20)$ and let $m$ increase from 10 to 500.    Table~\ref{coverage.table} shows the average coverage probabilities and widths, with standard errors in parentheses, across all slope coefficients based on 500 Monte Carlo simulations. 
 Across all the settings, the AHR-based confidence intervals are consistently accurate with tight width and reliable with high coverage. In the presence of heavy-tailed errors, the OLS-based confidence intervals tend to be wider, and standard errors of the interval width are also larger than those of the AHR method  by one order of magnitude.

\subsection{Distributed regularized Huber regression}

In the high-dimensional setting where the dimension $p$ exceeds the sample size $n$,  we compare four methods across a range of settings: (1) centralized $\ell_1$-penalized AHR estimator; (2) DC $\ell_1$-penalized AHR estimator; (3) centralized Lasso; and (4) distributed regularized AHR estimator with $T= \lfloor \log(m) \rfloor$ rounds of communication and with a local Lasso estimator as the initialization.  All four methods involve a regularization parameter $\lambda$,  which will be tuned by a held-out validation set  of size $\lfloor 0.25 N \rfloor$.  The robustification parameter $\tau$ in methods (1), (2) and (4) is selected by the self-tuning principle proposed by \cite{WZZZ2020}.   
 
The simulated data $\{ (y_i , x_i) \}_{i=1}^N$ is generated from a heteroscedastic model $y_i = x_i^\T\beta^* + c^{-1}(x_i^\T\beta^*)^2\varepsilon_i$,  where $\beta^* = (1.5, 1.5, 1.5, 1.5, 1.5, 0, \ldots, 0)^\T \in \mathbb R^p$, $x_i = (1, x_{i2}, \ldots, x_{ip})^\T$ with $x_{ij} \sim {\cN}(0, 1)$ for $j=2,\ldots, p$,  and $c = \sqrt{3} \|\beta^*\|_2^2$.   The regression errors $\varepsilon_i$ are generated from one of the four distributions considered in Section~\ref{numerical.low.dimension}, which are $\cN(0,1)$, $t_2$ (heavy-tailed and symmetric),  Par$(4,2)$ and Burr$(1, 2, 1)$ (heavy-tailed and skewed).
We fix $(n,p) = (250,1000)$ and let $m$ increase from 10 to 50. 
Figure~\ref{high_dim_estimation} plots the  $\ell_2$ error $\| \wh \beta - \beta^* \|_2$ versus the number of machines $m$,  averaged over 100 replications,  for all four methods.  The averaging $\ell_1$-penalized AHR estimator has a nondecaying estimation error as $m$ increases, which is expected because of its sub-optimal convergence rate that scales with the local sample size $n$.  The distributed AHR estimator with $T= \lfloor \log(m) \rfloor$ rounds of communication performs as good as the centralized AHR on the entire data set, and has much smaller estimation errors than the centralized Lasso in heavy-tailed cases.
Furthermore,  from the boxplots displayed in Figure~\ref{high_dim_box}  we see that the distributed AHR improves upon centralized Lasso in terms of both average performance and variability.

\begin{table}[!ht]
\begin{center}
\resizebox{\columnwidth}{!}{
\begin{tabular}{llcccccccc}
\hline
     &  & \multicolumn{2}{c}{$N(0,1)$}     & \multicolumn{2}{c}{$t_2$}     & \multicolumn{2}{c}{Par(4,2)}   & \multicolumn{2}{c}{Burr(1,2,1)} \\ \hline
     &  & \begin{tabular}[c]{@{}c@{}}Coverage\\ mean (sd)\end{tabular} & \begin{tabular}[c]{@{}c@{}}Width\\ mean (sd)\end{tabular} & \begin{tabular}[c]{@{}c@{}}Coverage\\ mean (sd)\end{tabular} & \begin{tabular}[c]{@{}c@{}}Width\\ mean (sd)\end{tabular} & \begin{tabular}[c]{@{}c@{}}Coverage\\ mean (sd)\end{tabular} & \begin{tabular}[c]{@{}c@{}}Width\\ mean (sd)\end{tabular} & \begin{tabular}[c]{@{}c@{}}Coverage\\ mean (sd)\end{tabular} & \begin{tabular}[c]{@{}c@{}}Width\\ mean (sd)\end{tabular} \\ \hline
\multirow{2}{*}{$m=50$}    & Dist-OLS  &0.93(0.011) & 0.029(0.001) & 0.93(0.011) & 0.097(0.056) & 0.93(0.012) & 0.35(0.420) & 0.94(0.011) & 0.088(0.068) \\
						     & Dist-AHR &0.95(0.007) & 0.031(0.001) & 0.95(0.009) & 0.077(0.007) & 0.95(0.008) & 0.23(0.025) & 0.95(0.009) & 0.058(0.006) \\
\multirow{2}{*}{$m=100$}  & Dist-OLS  &0.93(0.012) & 0.020(0.000) & 0.94(0.010) & 0.072(0.056) & 0.93(0.012) & 0.25(0.220) & 0.93(0.008) & 0.058(0.021) \\
						     & Dist-AHR &0.95(0.010) & 0.022(0.001) & 0.96(0.008) & 0.058(0.005) & 0.95(0.009) & 0.18(0.017) & 0.95(0.009) & 0.044(0.004) \\
\multirow{2}{*}{$m=200$}  & Dist-OLS  &0.93(0.011) & 0.014(0.000) & 0.93(0.013) & 0.052(0.031) & 0.93(0.010) & 0.18(0.095) & 0.94(0.015) & 0.044(0.021) \\
 						     & Dist-AHR &0.96(0.007) & 0.015(0.000) & 0.95(0.011) & 0.043(0.003) & 0.95(0.009) & 0.13(0.012) & 0.96(0.012) & 0.034(0.003) \\
\multirow{2}{*}{$m=300$}  & Dist-OLS  &0.93(0.013) & 0.012(0.000) & 0.94(0.011) & 0.043(0.022) & 0.94(0.011) & 0.18(0.820) & 0.93(0.008) & 0.038(0.020) \\
						     & Dist-AHR &0.95(0.010) & 0.013(0.000) & 0.96(0.010) & 0.036(0.003) & 0.95(0.012) & 0.11(0.009) & 0.96(0.009) & 0.028(0.002) \\
\multirow{2}{*}{$m=400$}  & Dist-OLS  &0.93(0.010) & 0.010(0.000) & 0.94(0.011) & 0.040(0.046) & 0.93(0.008) & 0.13(0.071) & 0.94(0.010) & 0.032(0.014) \\
 						     & Dist-AHR &0.95(0.009) & 0.011(0.000) & 0.96(0.009) & 0.031(0.002) & 0.95(0.012) & 0.10(0.008) & 0.96(0.009) & 0.025(0.002) \\
                       \hline
\end{tabular}
}
\caption{Coverage probabilities and widths (with standard errors in parentheses) of the normal-based CIs (averaged over all slope coefficients) for the distributed OLS and distributed AHR methods, based on 500 Monte Carlo simulations. } \label{coverage.table}
\end{center}
\end{table}

\begin{figure}[!htbp]
\begin{center}
\begin{tabular}{cccc}
\includegraphics[width=34mm]{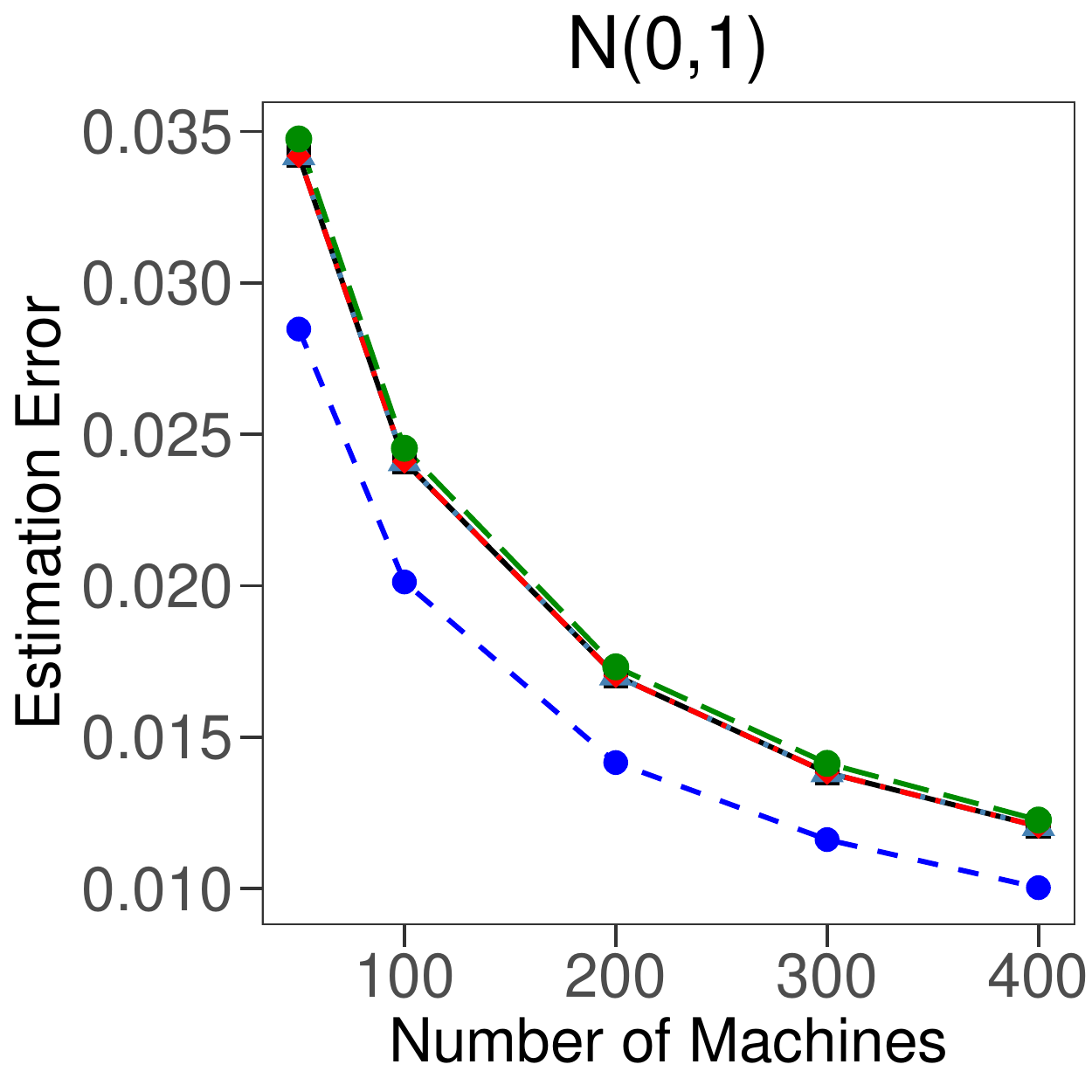} &   \includegraphics[width=34mm]{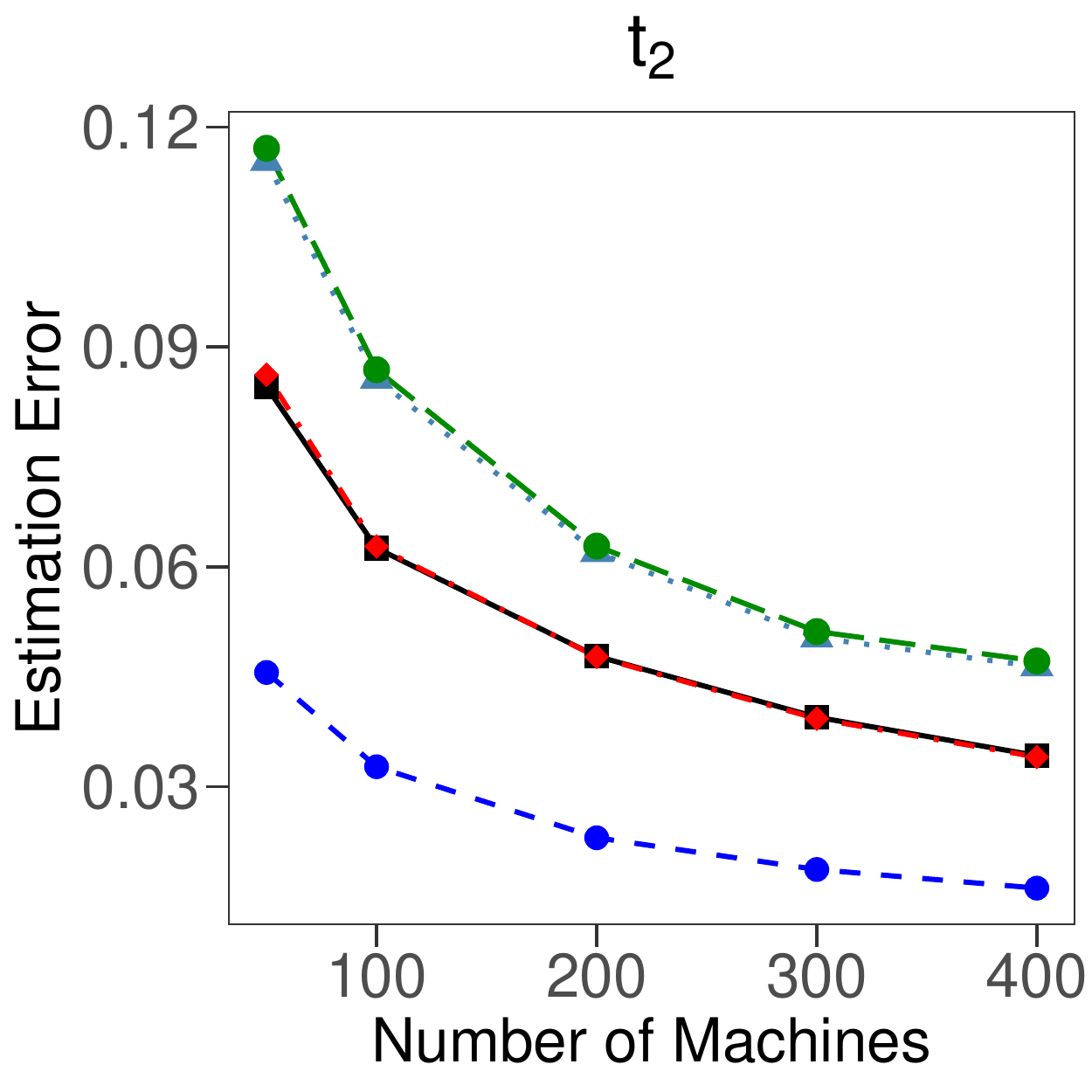} &  \includegraphics[width=34mm]{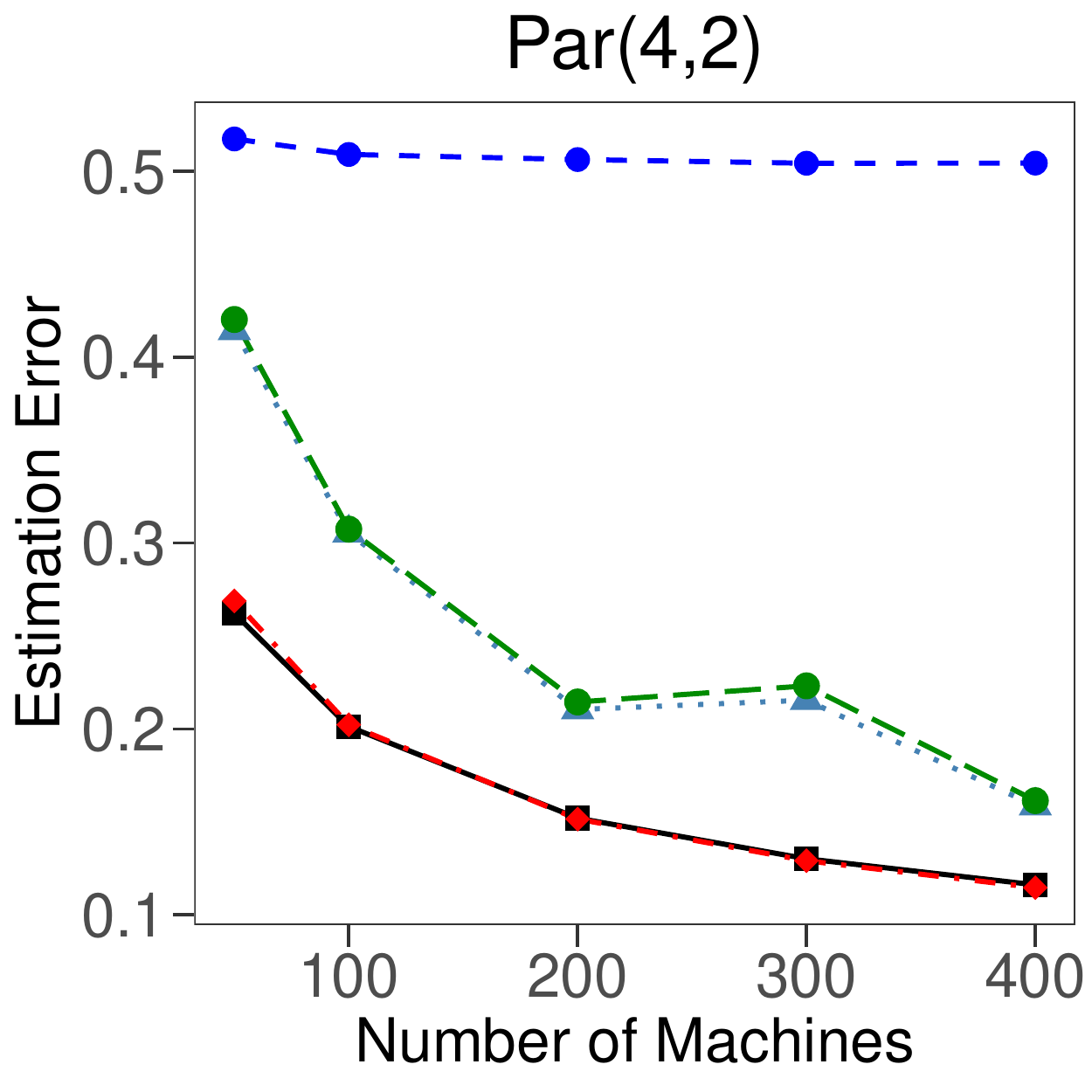} &   \includegraphics[width=34mm]{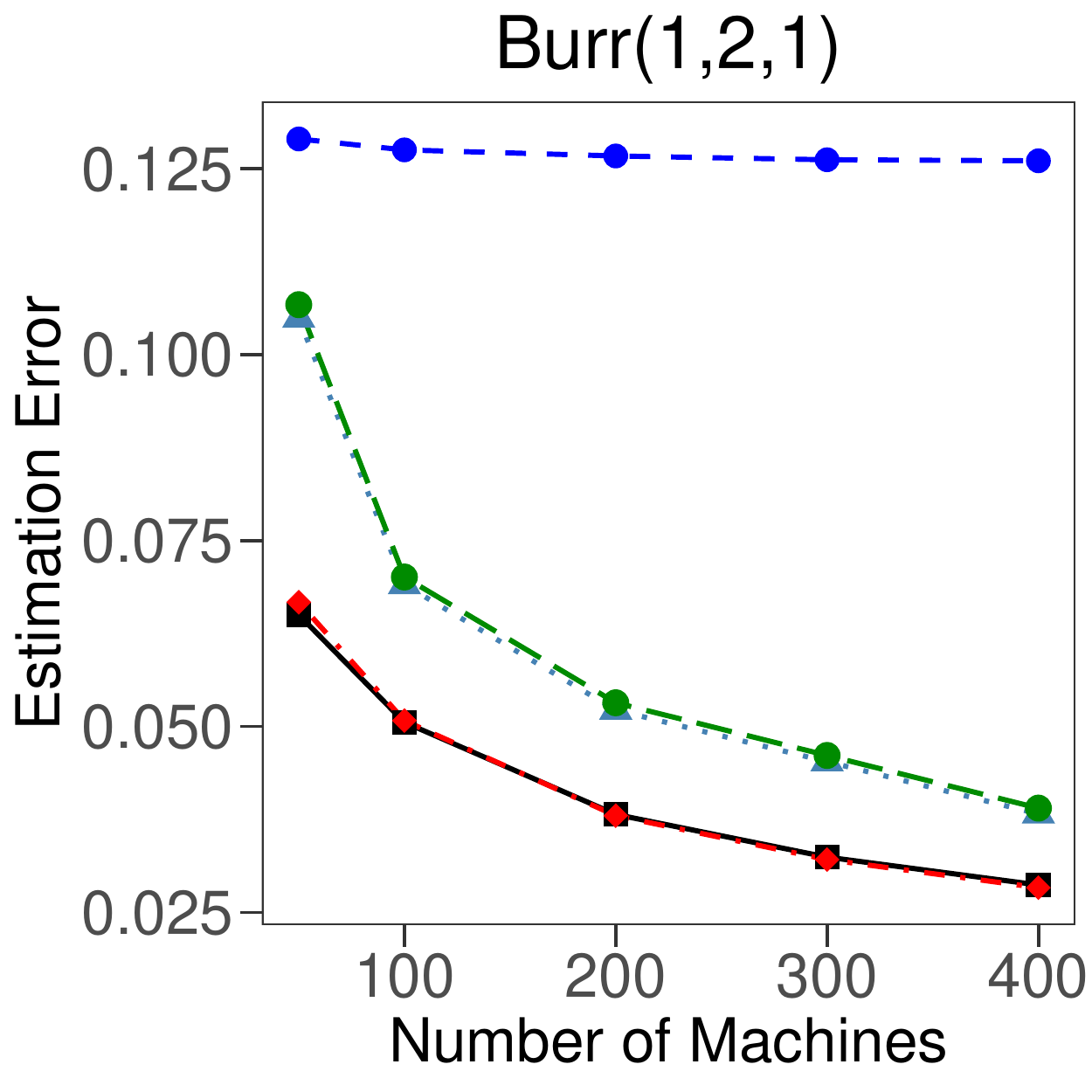}\\
(a) & (b) & (c) & (d)\\
\end{tabular}
\caption{Plots of estimation error (under $\ell_2$-norm) versus   number of machines when $(n,p)= (400, 20)$,  averaged over 500 replications.
Five estimators are presented:  global AHR estimator (\includegraphics[width= 8.5mm]{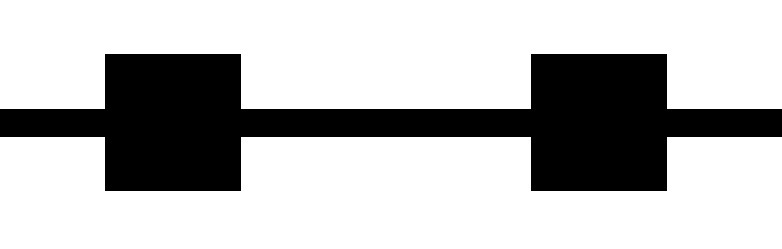}); DC-AHR estimator (\includegraphics[width= 8.5mm]{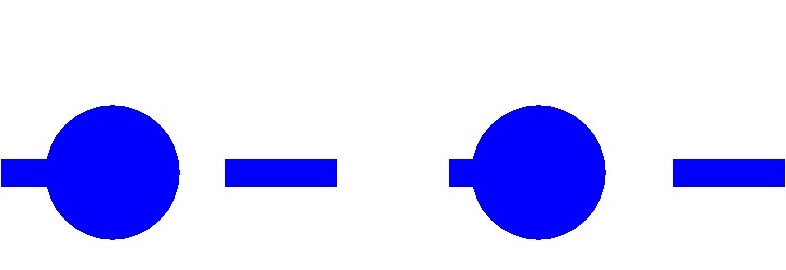}); DC-OLS estimator (\includegraphics[width= 8.5mm]{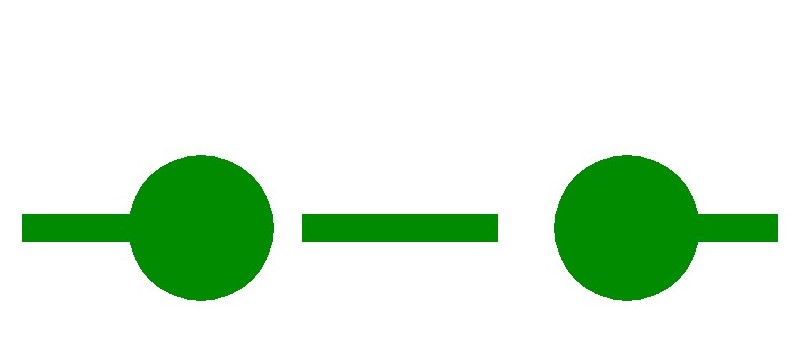});  distributed OLS estimator (\includegraphics[width= 8.5mm]{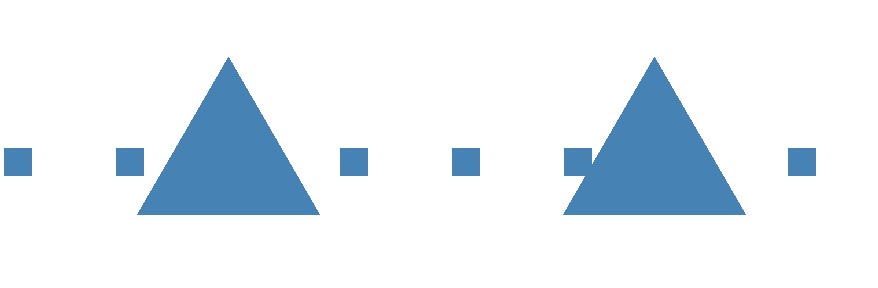}); and distributed AHR estimator (\includegraphics[width= 8.5mm]{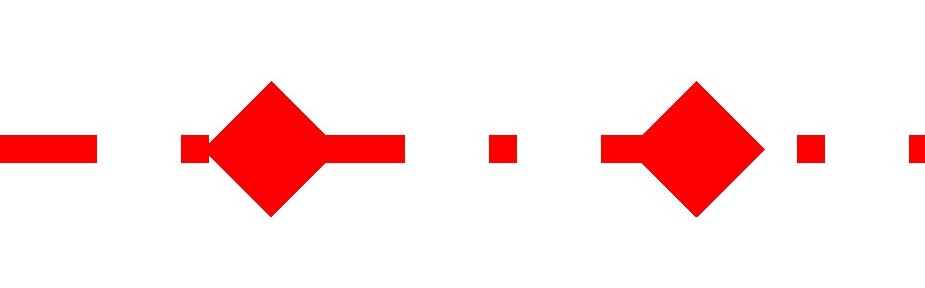}).}  \label{low_dim_estimation}
\end{center}
\end{figure}

\begin{figure}[!htbp]
\begin{center}
\begin{tabular}{cccc}
\includegraphics[width=34mm]{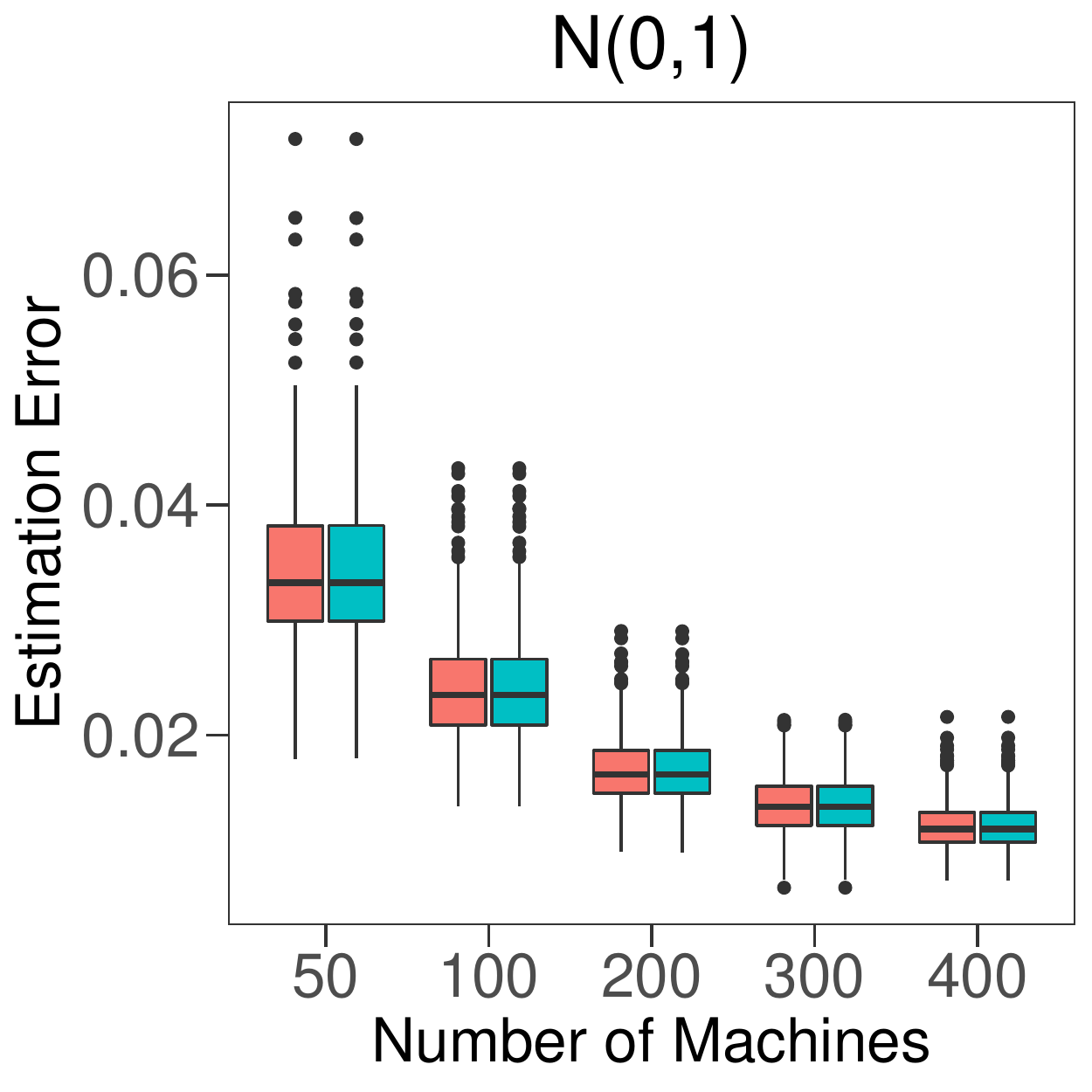} &   \includegraphics[width=34mm]{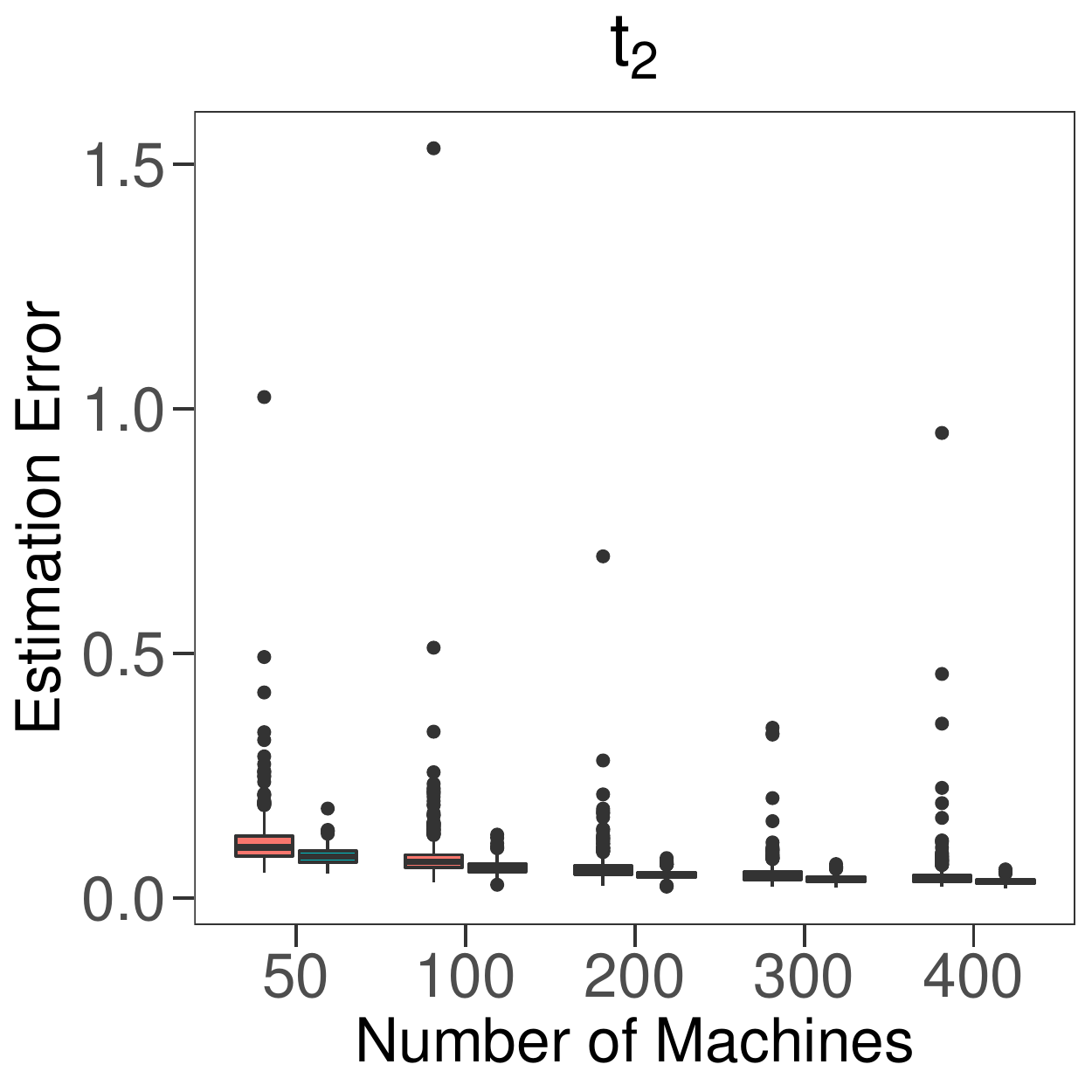} &  \includegraphics[width=34mm]{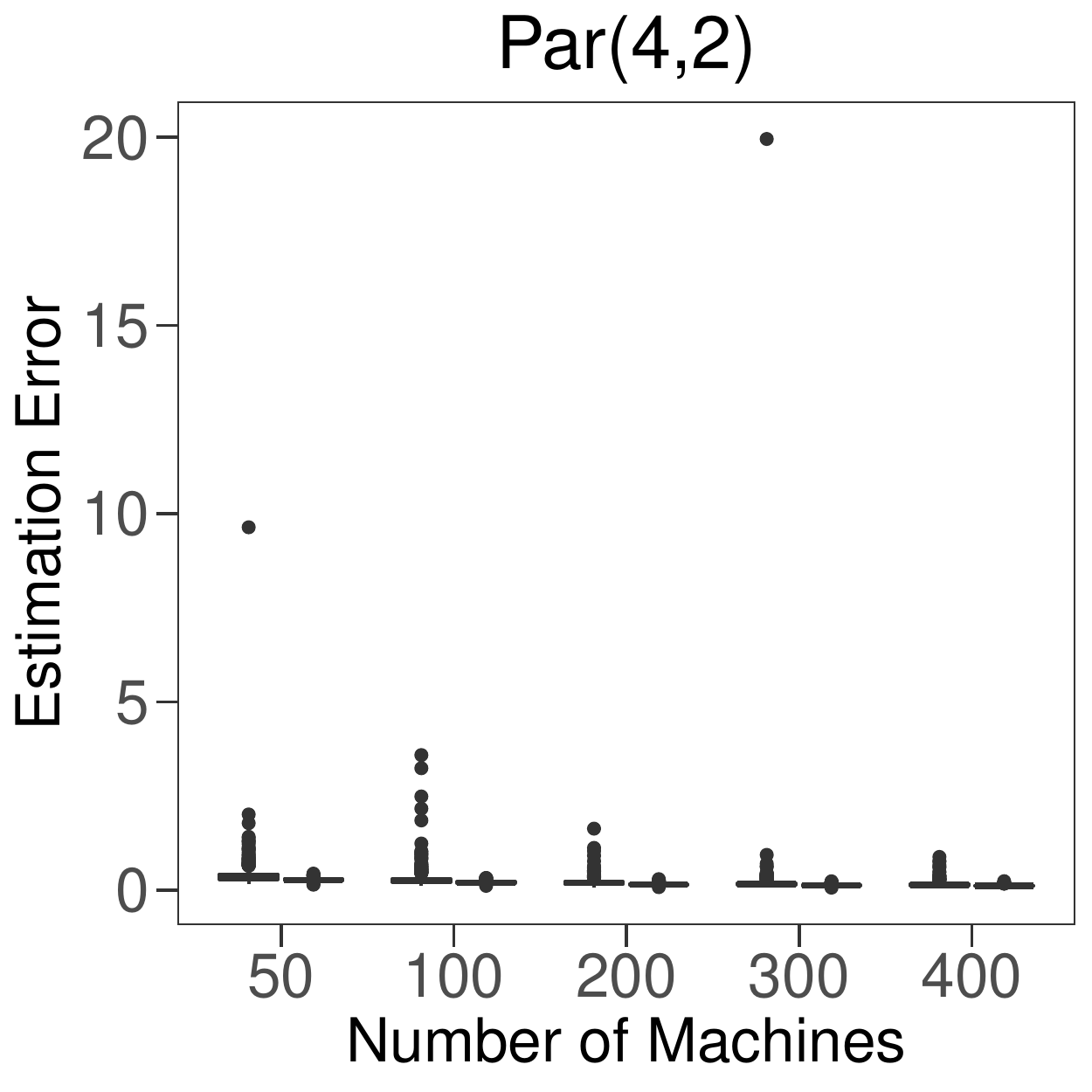} &   \includegraphics[width=34mm]{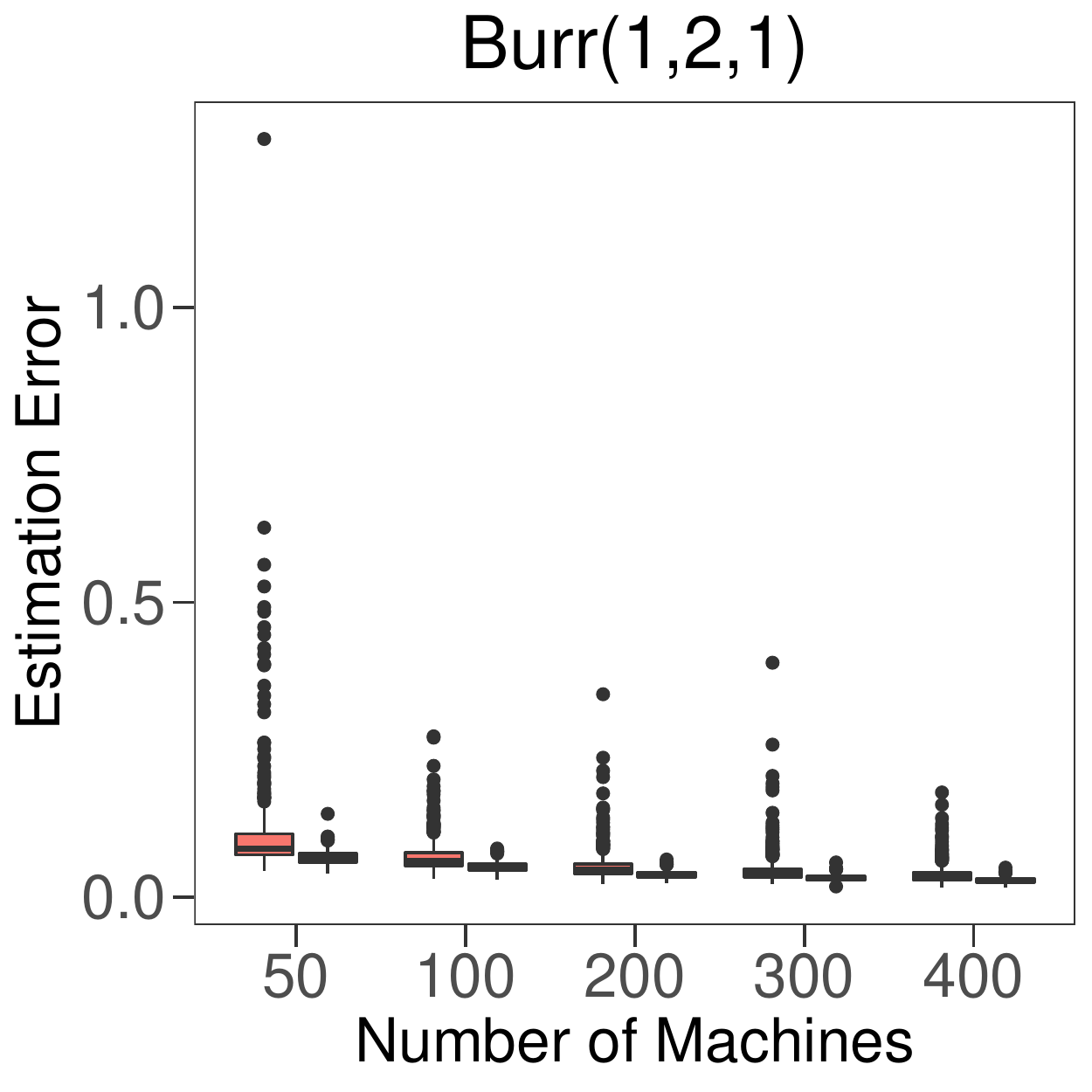}\\
(a) & (b) & (c) & (d)\\
\end{tabular}
\caption{Boxplots of estimation error (under $\ell_2$-norm) versus the  number of machines when $(n,p)= (400, 20)$ for distributed OLS estimator (\includegraphics[width= 1.53mm]{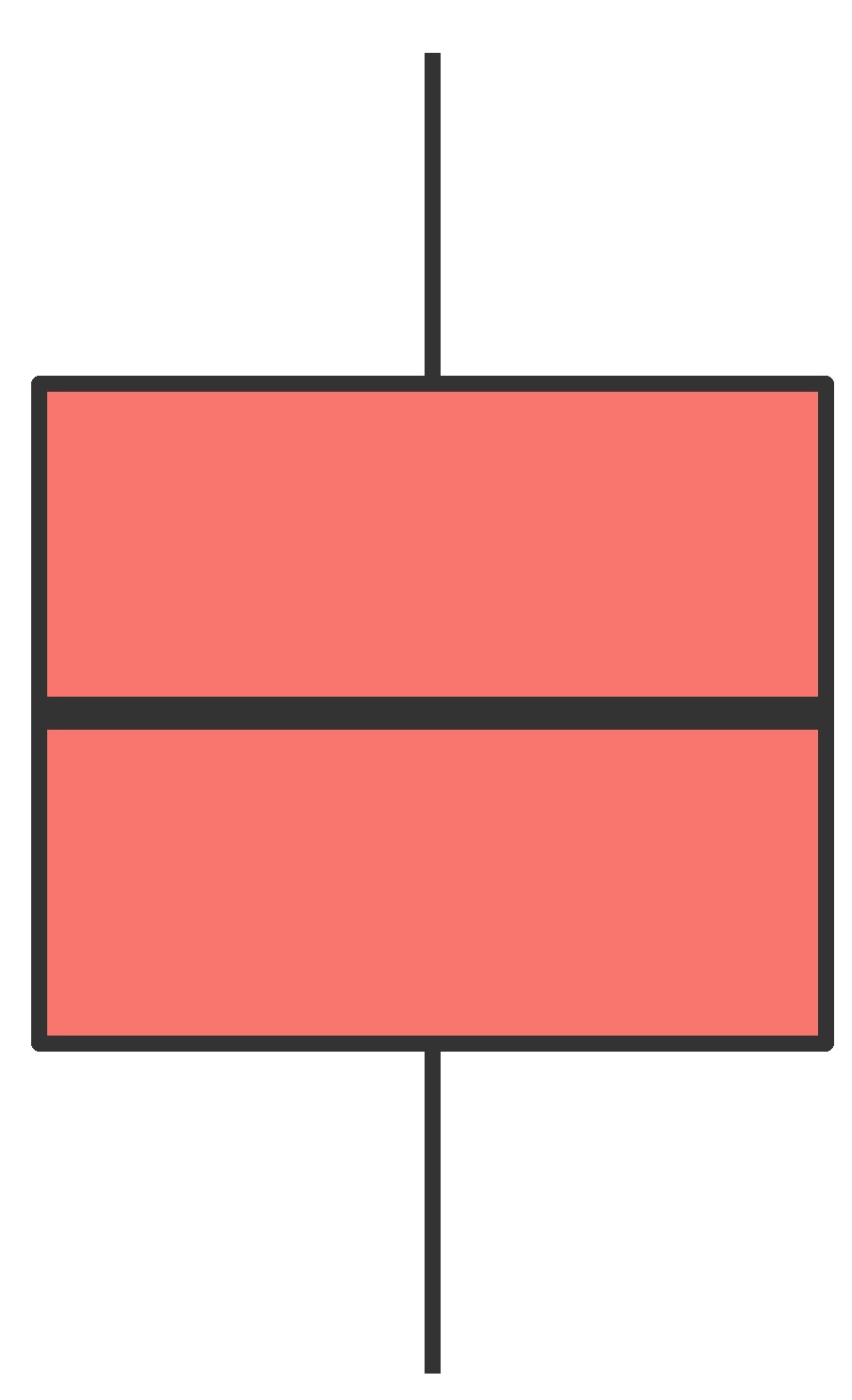}) and distributed AHR estimator (\includegraphics[width= 1.53mm]{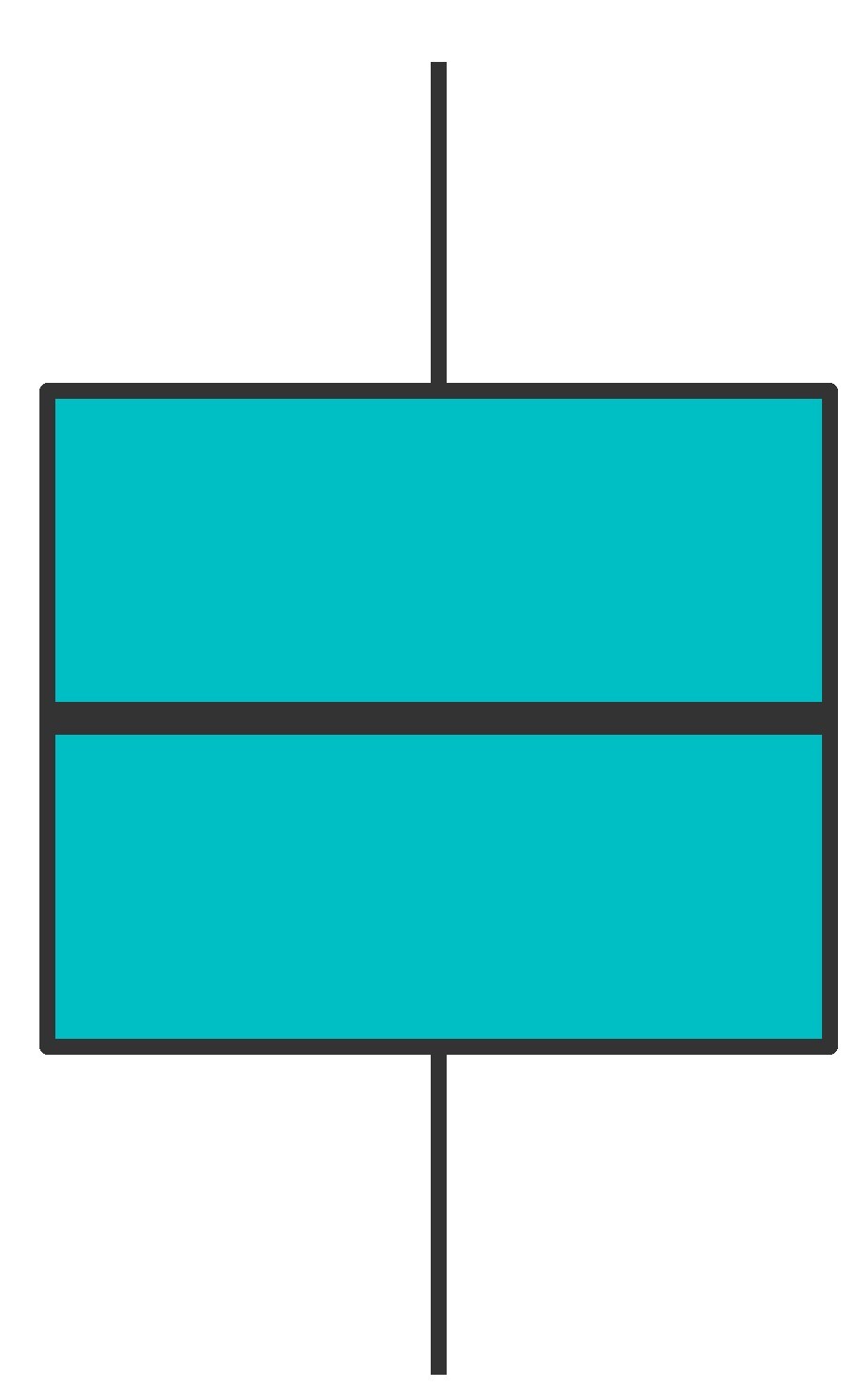}), averaged over 500 replications. }  \label{low_dim_box}
\end{center}
\end{figure}

\begin{figure}[!htbp]
\begin{center}
\begin{tabular}{cccc}
\includegraphics[width=34mm]{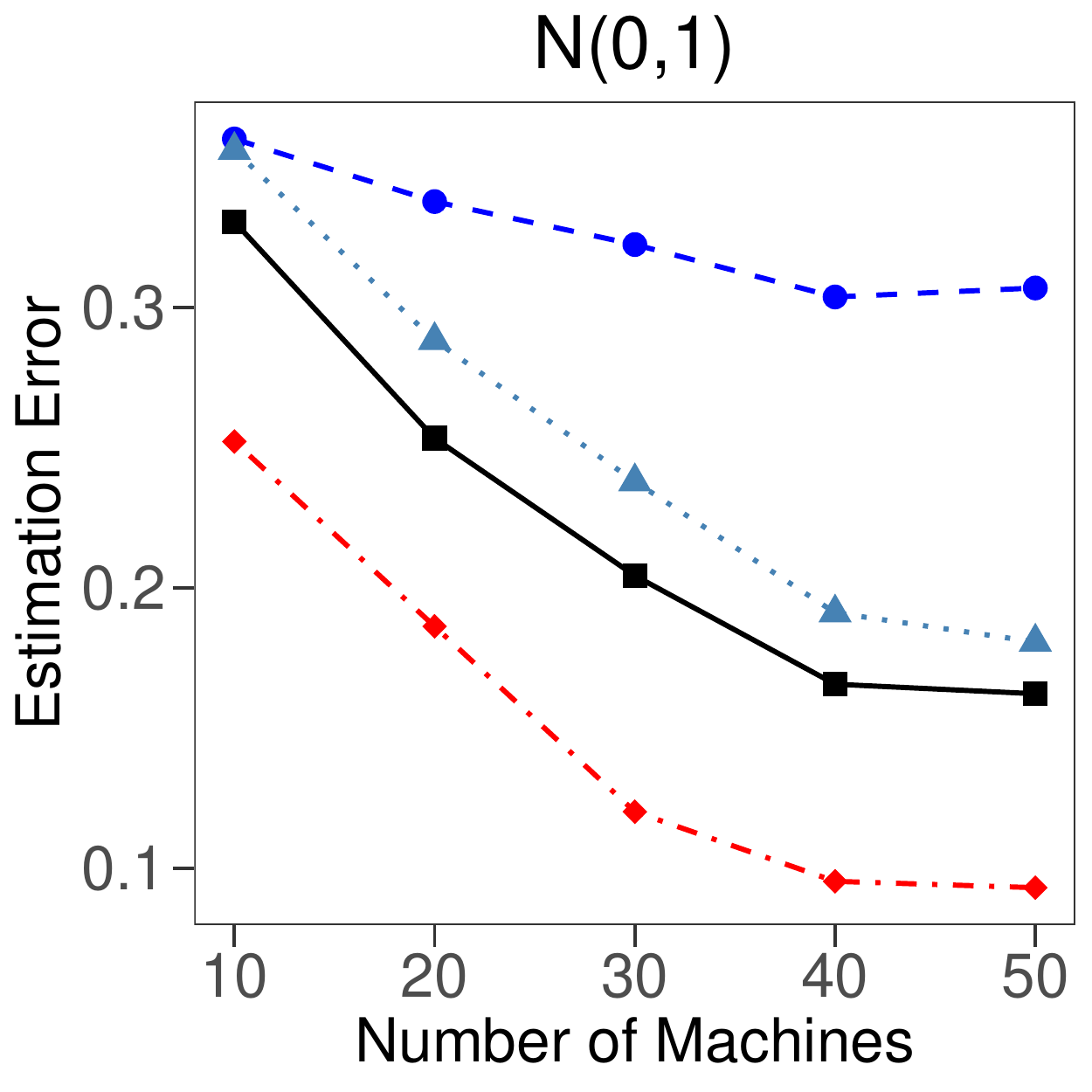} &   \includegraphics[width=34mm]{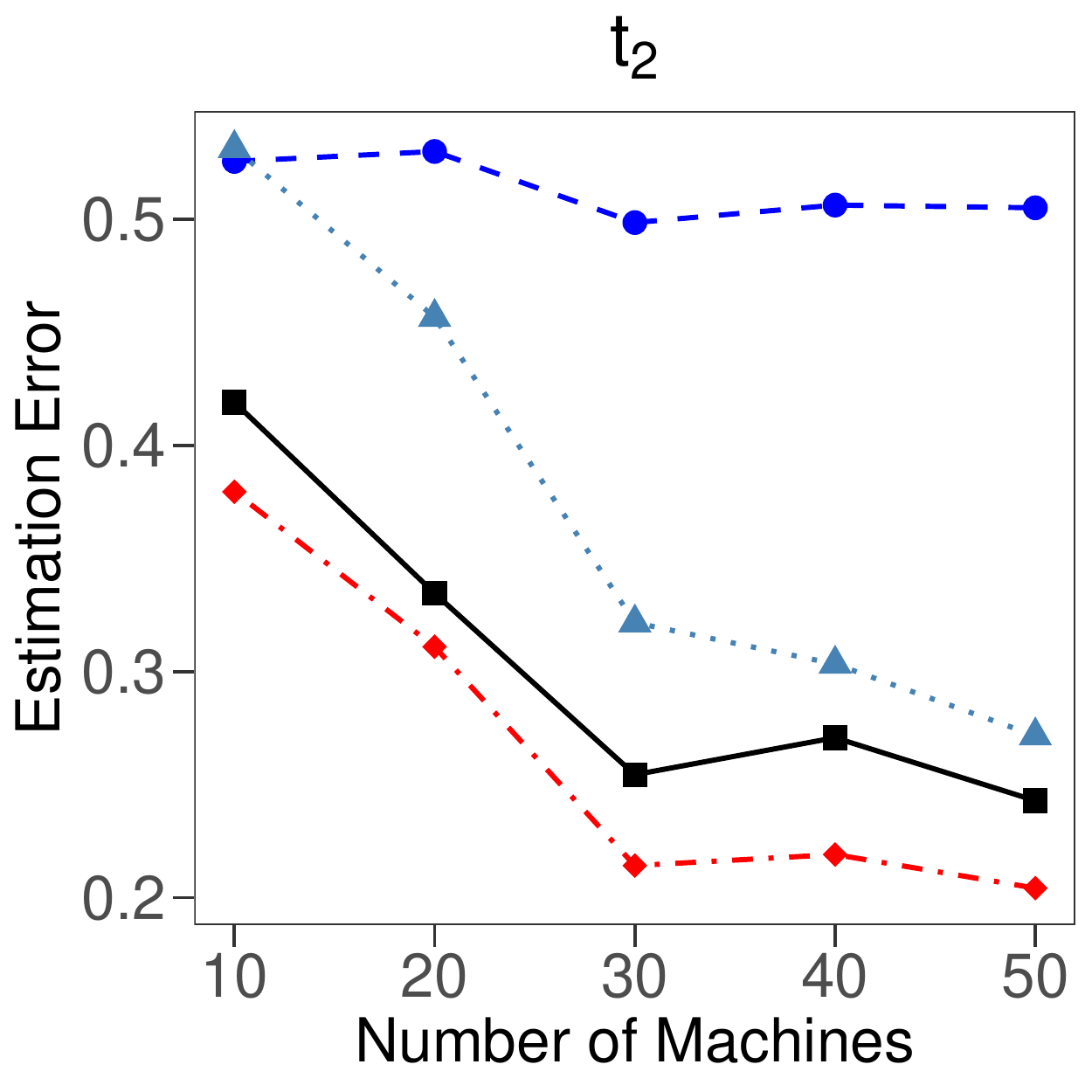} &  \includegraphics[width=34mm]{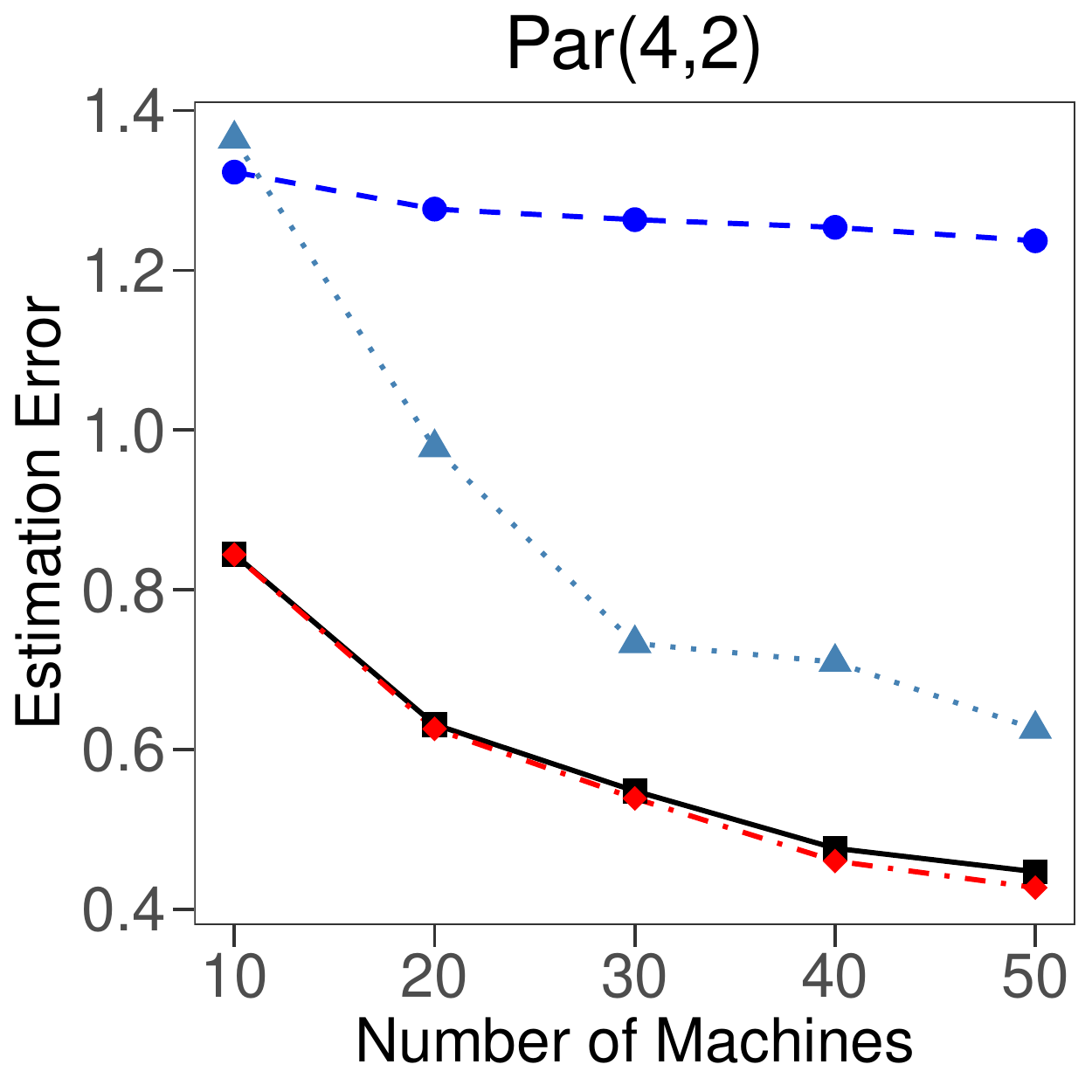} &   \includegraphics[width=34mm]{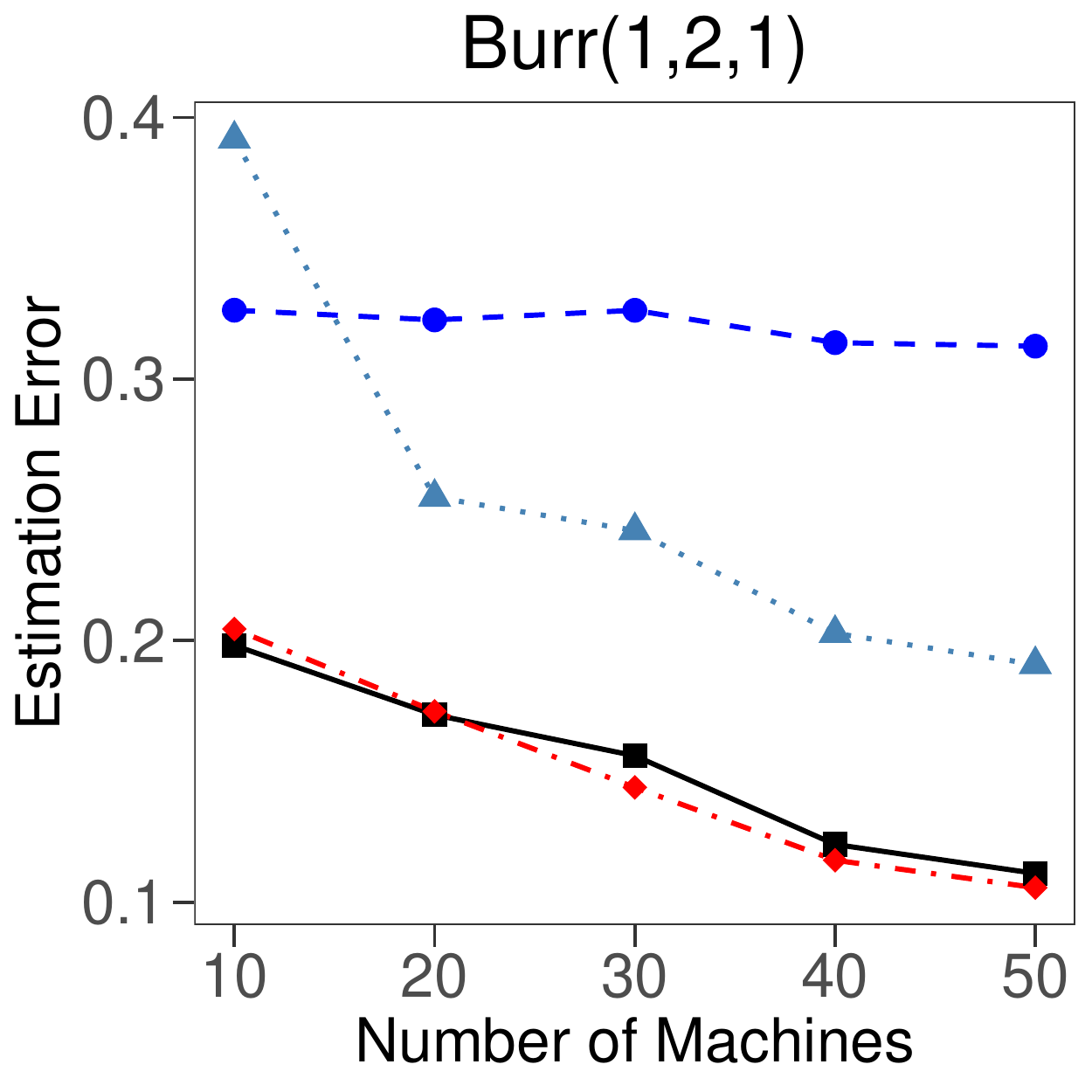}\\
(a) & (b) & (c) & (d)\\
\end{tabular}
\caption{Plots of estimation error (under $\ell_2$-norm) versus the number of machines,  over 100 replications,  under a high-dimensional heteroscedastic model when $(n,p,s)= (250,1000,5)$. 
Four estimators are presented: centralized $\ell_1$-penalized AHR estimator (\includegraphics[width= 8.5mm]{figures/lineblack.jpg}); DC $\ell_1$-penalized AHR estimator (\includegraphics[width= 8.5mm]{figures/lineblue.jpg}); centralized Lasso estimator (\includegraphics[width= 8.5mm]{figures/linesteel.jpg}); and proposed distributed regularized AHR estimator (\includegraphics[width= 8.5mm]{figures/linered.jpg})}  \label{high_dim_estimation}
\end{center}
\end{figure}

\begin{figure}[!htbp]
\begin{center}
\begin{tabular}{cccc}
\includegraphics[width=34mm]{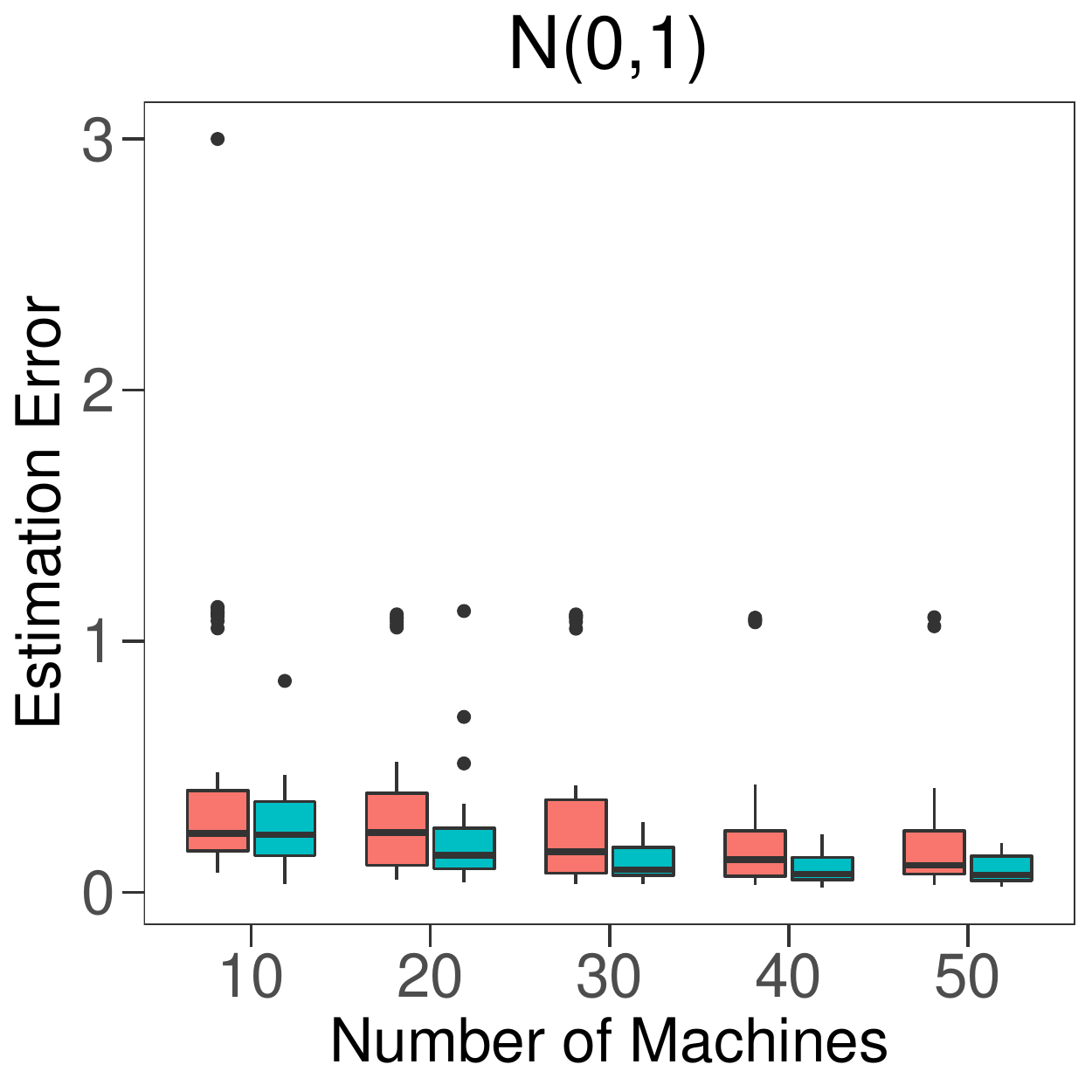} &   \includegraphics[width=34mm]{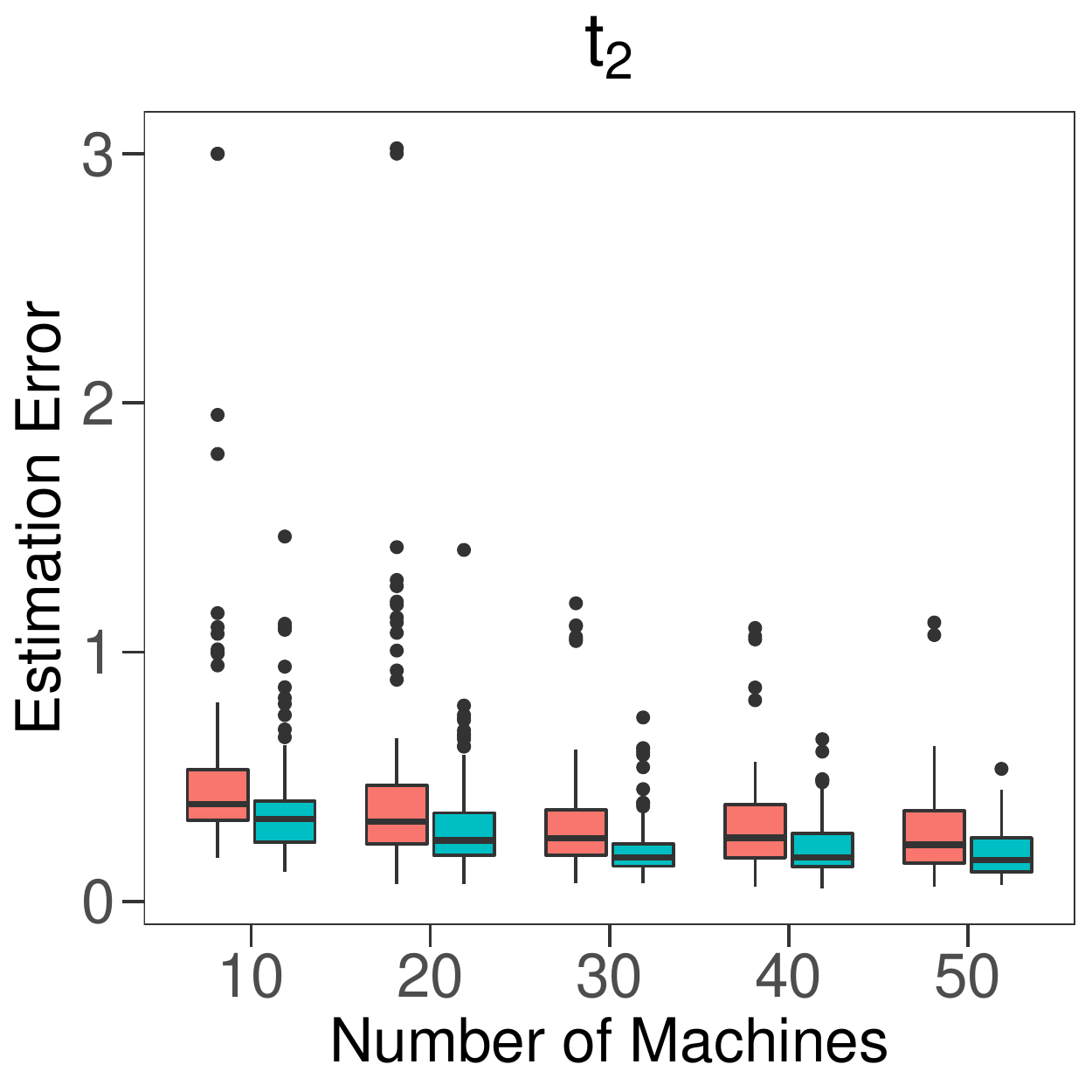} &  \includegraphics[width=34mm]{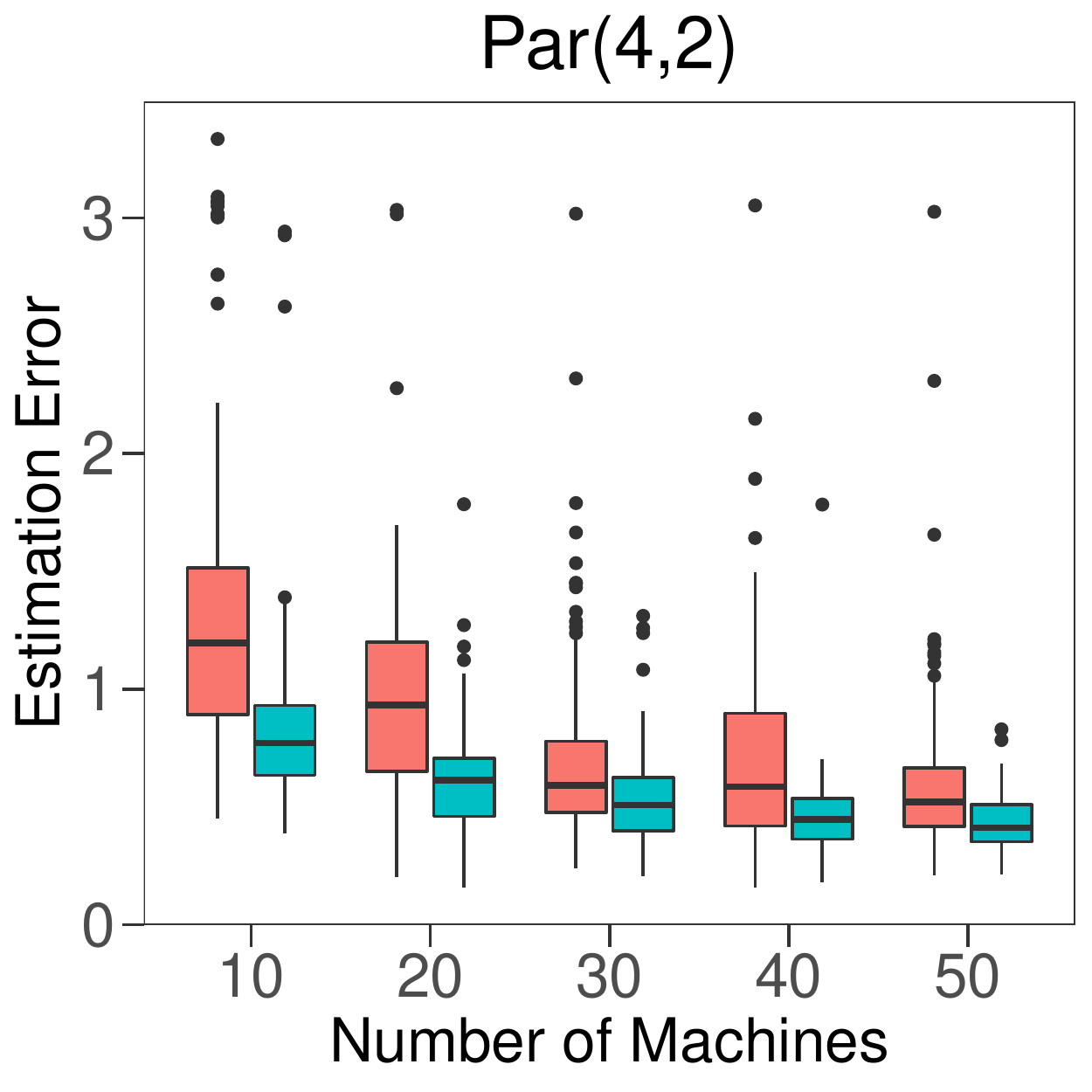} &   \includegraphics[width=34mm]{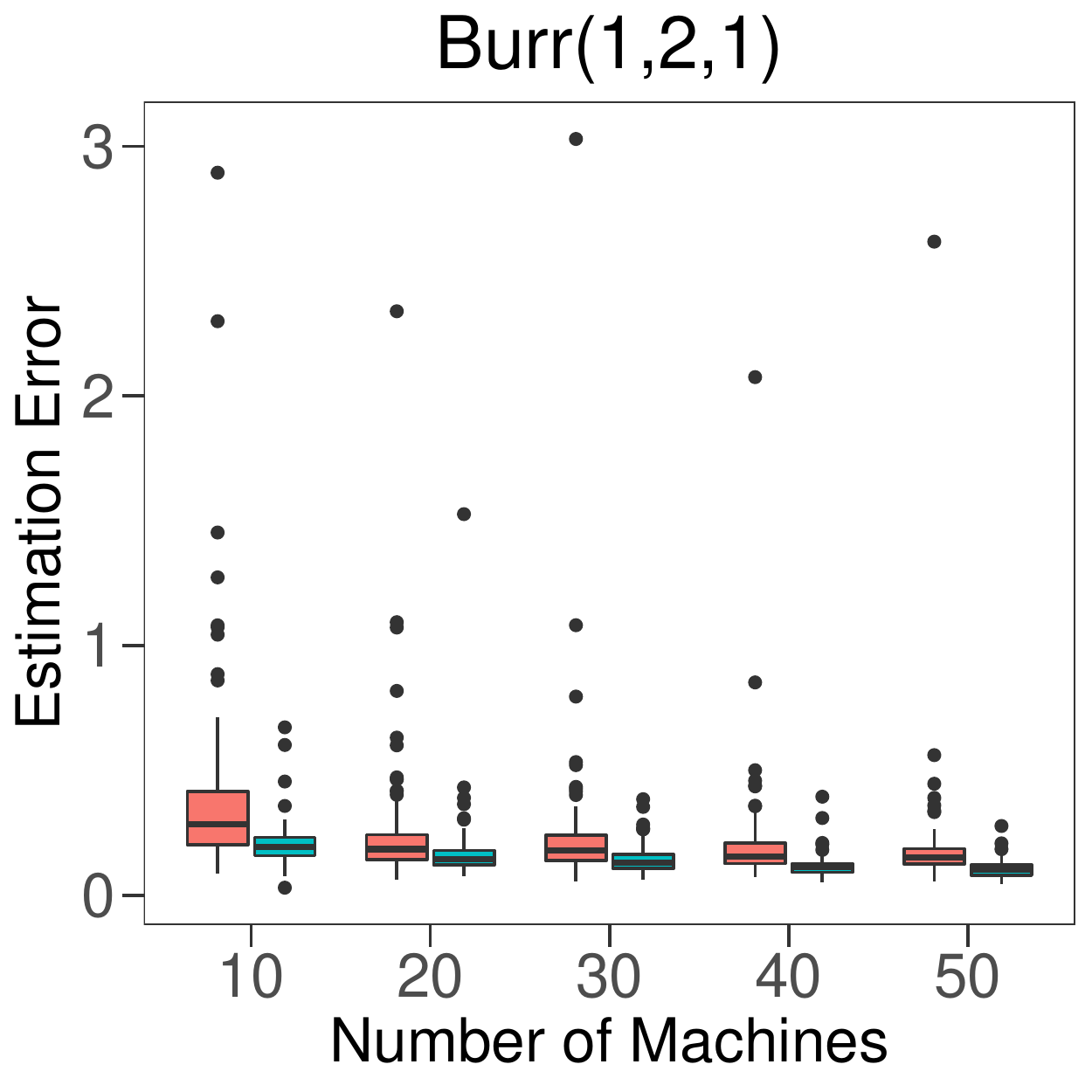}\\
(a) & (b) & (c) & (d)\\
\end{tabular}
\caption{Boxplots of estimation errors (under $\ell_2$-norm) versus the number of machines,  over 100 replications,  for  centralized Lasso  (\includegraphics[width= 1.53mm]{figures/box_ols.jpg}) and distributed AHR (\includegraphics[width= 1.53mm]{figures/box_dist.jpg}) under a high-dimensional heteroscedastic model when $(n,p,s)= (250,1000,5)$. }  \label{high_dim_box}
\end{center}
\end{figure}

\newpage
{\textbf{\Large Appendix} 

\appendix

\section{Preliminaries}

For any convex function $\psi : \RR^k \to \RR$, define the corresponding Bregman divergence $D_{\psi}(w', w) = \psi(w') - \psi(w) - \langle \nabla \psi(w), w' - w \rangle$ and its symmetrized version
\#
	\overbar D_{\psi}(w, w') = D_{\psi}(w , w') + D_{\psi}(w', w)  = \big\langle \nabla \psi(w) - \nabla \psi(w'), w - w' \big\rangle, \ \ w, w' \in \RR^k.
\#

Let $z = \Sigma^{-1/2} x \in \RR^p$ be the standardized vector of covariates such that $\EE (z z^\T ) = I_p$, and define 
$\mu_k = \sup_{u \in \mathbb{S}^{p-1}} \EE  |z^\T u |^k$ for $k\geq 1$. In particular, $\mu_2=1$.
 For every $\delta \in (0, 1]$, define
\#
	\eta_\delta = \inf\Bigg\{ \eta>0 :  \sup_{u \in \mathbb{S}^{p-1}} \EE \big\{ (z^\T u )^2 \mathbbm{1}(|z^\T u |> \eta )  \big\} \leq \delta  \Bigg\}. \label{bregman.def}
\#

Under Condition~(C1), $\eta_\delta$ depends only on $\delta$ and $\upsilon_1$, and the map $\delta \mapsto \eta_\delta$ is non-increasing with $\eta_\delta \downarrow 0$ as $\delta \to 1$. A crude bound for $\eta_\delta$, as a function of $\delta$, is  $\eta_\delta \leq  (\mu_4 / \delta)^{1/2}$.

In Lemmas~\ref{lem:local.RSC} and \ref{lem:global.score} below, we provide a lower bound on the symmetrized Bregman divergence and an upper bound on the score, respectively. The former is a direct consequence of Lemmas~C.3 and C.4 in \cite{SZF2020} with slight modifications, and the latter combines Lemmas~C.5 and C.6 in \cite{SZF2020} with $\delta=1$.
For the shifted Huber loss $\wt \cL(\cdot)$, note that
\$
 \overbar D_{\wt \cL}(\beta, \beta^* ) = \bigl\langle \nabla \wh \cL_{1,\kappa}(\beta) - \nabla \wh \cL_{1,\kappa} (\beta^*), \beta - \beta^* \bigr\rangle.
\$
 Moreover, define the $\ell_1$-cone 
\$
	\Lambda = \bigl\{ \beta \in\RR^p : \| \beta - \beta^*\|_1 \leq  4 s^{1/2} \| \beta - \beta^* \|_{\Sigma}  \bigr\}.
\$

\begin{lemma}   \label{lem:local.RSC}
Let $\kappa, r>0$ satisfy $\kappa  \geq 4 \max( \eta_{0.25} r, \sigma)$. 
\begin{itemize}
\item[(i)] Condition~(C1) ensures that, with probability at least $1-e^{-u}$,
\#   \label{rsc.lbd}
	 \overbar D_{\wt \cL}(\beta, \beta^* ) \geq \frac{1}{4} \| \beta - \beta^* \|_{\Sigma}^2 ~\mbox{ holds uniformly over }~ \beta \in \Theta(r) 
\#
as long as $n\gtrsim (\kappa /r)^2 (p+u)$.

\item[(ii)] Condition~(C2) ensures that, with probability at least $1-e^{-u}$,
\#   \label{hd-rsc.lbd}
	 \overbar D_{\wt \cL}(\beta, \beta^* ) \geq \frac{1}{4} \| \beta - \beta^* \|_{\Sigma}^2 ~\mbox{ holds uniformly over }~ \beta \in \Theta(r) \cap \Lambda
\#
as long as $n\gtrsim   (  \kappa/ r)^2 (s \log p + u)$.
\end{itemize}

\begin{proof} 
Without loss of generality, assume $\cI_1 = \{ 1, \ldots, n\}$.
It suffices to prove \eqref{hd-rsc.lbd} under Condition~(C2). Following the proof of Lemma~C.4 in \cite{SZF2020}, the key is to upper bound the expected value of the maximum $\| (1/n) \sn e_i x_i  \|_\infty$, where $e_1,\ldots, e_n$ are independent Rademacher random variables. Let $\EE _e$ be the expectation with respect to $e_1,\ldots,e_n$ conditional on the remaining variables. By Hoeffding's moment inequality,
\$
	\EE _e \Bigg\| \frac{1}{n} \sn e_i x_i \Bigg\|_\infty \leq \max_{1\leq j\leq p}  \Biggl( \frac{1}{n} \sn x_{ij}^2 \Biggr)^{1/2} \sqrt{\frac{2\log(2p)}{n}} \leq B \sqrt{\frac{2\log(2p)}{n}} ,
\$
which in turns implies $\EE  \| (1/n) \sn e_i x_i  \|_\infty \leq B \sqrt{2\log(2p)/n}$. Keep the rest of the proof the same proves the claimed bound.
\end{proof}

\end{lemma}

Consider the gradient $\nabla \wh \cL_\tau(\cdot)$ evaluated at $\beta^*$, namely,
$$
	\nabla \wh  \cL_\tau(\beta^*) =  - \frac{1}{N} \sum_{i=1}^N  \psi_\tau(\varepsilon_i) x_i, 
$$
where $\psi_\tau(u)=\ell_\tau'(u)$. The following lemma provides high probability bounds on both $\ell_2$- and $\ell_\infty$-norms of $\nabla \wh \cL_\tau(\beta^*) $. Recall that $\Omega = \Sigma^{-1}$.

\begin{lemma} \label{lem:global.score}
Let $u>0$ and write $\cL_\tau(\cdot) = \EE  \wh \cL_\tau (\cdot)$.
\begin{itemize}
\item[(i)]  Condition~(C1) ensures that, with probability at least $1-e^{-u}$,
\#
	 \|   \nabla \wh  \cL_\tau(\beta^* ) - \nabla   \cL_\tau(\beta^* ) \|_{\Sigma^{-1}} \leq C_0 \Bigl\{   \sigma \sqrt{(p+u)/N}   + \tau(p+u)/N \Bigr\}   , \label{global.score.bound}
\#
where  $C_0>0$ is a constant  depending only on $ \upsilon_1$. Moreover, $\| \nabla   \cL_\tau(\beta^* ) \|_{\Omega} \leq \sigma^2/\tau$.

\item[(ii)]    Condition~(C2) ensures that, with probability at least $1- e^{-u}$,
\#
	 \|   \nabla \wh  \cL_\tau(\beta^* )  - \nabla   \cL_\tau(\beta^* ) \|_\infty  \leq  \sigma\sigma_u \sqrt{\frac{2 \{\log(2p) + u \}}{N}} + \frac{B \tau}{3} \frac{\log(2p)+u}{N} . \label{global.score.maxbound}
\#
 
\end{itemize}
\end{lemma}

\begin{proof}
The bound \eqref{global.score.bound} is an immediate consequence of Lemma~C.3 in \cite{SZF2020}.
It suffices to prove \eqref{global.score.maxbound} under Condition~(C2). Note that 
\$
 \|   \nabla \wh  \cL_\tau(\beta^* )  - \nabla   \cL_\tau(\beta^* ) \|_\infty  = \max_{1\leq j\leq p} \Bigg|  \frac{1}{N} \sum_{i=1}^N (1-\EE )  \xi_i x_{ij} \Bigg|,
\$
where $\xi_i := \psi_\tau(\varepsilon_i)$ satisfy $|\xi_i|\leq \tau$ and $\EE ( \xi_i^2 |x_i ) \leq \EE  (\varepsilon_i^2 | x_i) \leq \sigma^2$. For any $1\leq j\leq p$ and $z\geq 0$, applying Bernstein's inequality yields that with probability at least $1-2 e^{-z}$,
\$
\Bigg|  \frac{1}{N} \sum_{i=1}^N (1-\EE )  \xi_i x_{ij} \Bigg| \leq \sigma_{jj}^{1/2} \sigma \sqrt{\frac{2 z}{N}} + \frac{B \tau  }{3 } \frac{z}{N} .
\$
Taking $z=\log(2p) +u$, the claimed bound \eqref{global.score.maxbound} then follows from the union bound.
\end{proof}

\section{Proof of Main Results}

\subsection{Proof of Theorem~\ref{thm:one-step}}

\noindent
{\sc Proof of \eqref{one-step.error}}.
For simplicity, we write $\wt \beta = \wt \beta^{(1)}$, which minimizes the shifted Huber loss $\wt \cL(\cdot)$ and thus satisfies the first-order condition $\nabla \wt\cL(\wt \beta) = 0$. Throughout the proof we assume the event $\cE_0(r_0) \cap \cE_*(r_*)$ occurs. In view of Lemma~\ref{lem:local.RSC}, we consider a local region $\Theta(r_{\loc})$ with $r_{\loc}= \kappa/(4 \eta_{0.25})$, and define an intermediate estimator $\wt \beta_c = (1-c) \beta^* +c  \wt \beta$, where
\$
	 c := \sup\big\{ u\in [0,1] : (1-u) \beta^* +u  \wt \beta \in \Theta(r_{\loc}) \big\} \begin{cases}
	 =1 & {\rm if }~ \wt \beta \in \Theta(r_{\loc}) , \\
	 \in (0,1)  & {\rm otherwise}.
	 \end{cases}
\$
By construction, $\wt \beta_c \in \Theta(r_{\loc})$. In particular, if $ \wt \beta \notin \Theta(r_{\loc})$, we must have $\wt \beta_c$ lying on the boundary of $\Theta(r_{\loc})$, i.e. $\| \wt \beta_c - \beta^* \|_{\Sigma} =r_{\loc}$.


Applying Lemma~C.1 in \cite{SZF2020}, we see that the three  points $\wt \beta, \wt \beta_c$ and $\beta^*$ satisfy
$\overbar D_{\wt \cL}(\wt \beta_c , \beta^*) \leq c \overbar D_{\wt \cL}(\wt \beta,\beta^*)$, where  $\overbar D_{\wt \cL}( \beta , \beta^*) =  \langle \nabla \wt \cL(\beta) - \nabla  \wt \cL(\beta^*),    \beta -\beta^*  \rangle  = \langle  \nabla \wh \cL_{1,\kappa}(\beta) - \nabla \wh  \cL_{1,\kappa}(\beta^*) ,  \beta - \beta^*   \rangle$.
Together with the first-order condition $\nabla \wt \cL(\wt \beta) = 0$, this implies
\#
	\overbar D_{\wt \cL }(\wt \beta_c , \beta^* ) \leq - c  \big\langle \nabla \wt \cL(\beta^* ) , \wt \beta - \beta^* \big\rangle \leq  \|  \nabla\wt \cL(\beta^* )  \|_{\Omega} \cdot  \| \wt \beta_c  - \beta^* \|_{\Sigma} . \label{FOC}
\#
For the left-hand side of \eqref{FOC}, applying Lemma~\ref{lem:local.RSC} with $r=r_{\loc}$ and the fact  $\wt \beta_c \in \Theta(r_{\loc})$ yields that with probability at least $1-e^{-u}$,
\#
\overbar D_{\wt \cL}(\wt \beta_c , \beta^* ) \geq \frac{1}{4}     \| \wt \beta_c  - \beta^* \|_{\Sigma}^2 \label{local.rsc.bound}
\#
 as long as  $n\gtrsim  p+u $.

To bound the right-hand side of \eqref{FOC}, we  define vector-valued random processes
\#
\left\{\begin{array}{ll}
	\Delta_1(\beta)  = \Sigma^{-1/2}  \big\{ \nabla \wh \cL_{1,\kappa}(\beta) - \nabla \wh \cL_{1,\kappa}( \beta^*) \big\} - \Sigma^{1/2} ( \beta - \beta^* ) ,  \vspace{0.2cm}\\
~\,\Delta(\beta)  = \Sigma^{-1/2}  \big\{ \nabla \wh \cL_\tau(\beta) - \nabla \wh \cL_\tau( \beta^*) \big\} - \Sigma^{1/2} ( \beta - \beta^* ),
\end{array}\right. \label{def:Delta}
\#
Let $0<r_0\leq  \sigma$. Following the proof of Theorem~B.1 in the supplement of \cite{SZF2020}, it can be similarly shown that, with probability at least $1-2e^{-u}$,
\#
 \sup_{\beta \in \Theta(r_0)}  \| \Delta_1(\beta)  \|_2  & \leq C_1  \left(   \sqrt{\frac{p+u}{n}} + \frac{\sigma^2}{\kappa^2} \right) r_0   ~\mbox{ and }~ \sup_{\beta \in \Theta(r_0)}  \| \Delta(\beta)   \|_2    \leq     C_1 \left(   \sqrt{\frac{p+u}{N}} + \frac{\sigma^2}{\tau^2} \right) r_0   \label{diff.grad.unif.bound}
\#
as long as $n\gtrsim p+u$, where $C_1>0$ is a constant depending only on $\upsilon_1$. Recall that $\tau \geq \kappa \asymp \sigma\sqrt{n/(p+u)}$. Conditioned on event $\cE_0(r_0) \cap  \cE_*(r_*)$, it follows that
\#
 \|  \nabla \wt \cL (\beta^* )  \|_{\Omega}  &  = \|     \Delta(\wt \beta^{(0)}) - \Delta_1(\wt \beta^{(0)})  + \Sigma^{-1/2} \nabla \wh \cL_\tau(\beta^* )   \|_2 \nn \\
& \leq \|     \Delta(\wt \beta^{(0)}) - \Delta_1(\wt \beta^{(0)})   \|_2 +   \|  \nabla  \wh \cL_\tau(\beta^* ) \|_{\Omega} \nn\\
& \leq C_2     r_0 \sqrt{\frac{p+u}{n}}       +   r_*  . \label{surrogate.score.bound}
\#

Together, the bounds \eqref{FOC}, \eqref{local.rsc.bound} and \eqref{surrogate.score.bound} imply that, conditioning on $\cE_0(r_0) \cap \cE_*(r_*)$,
\#
	    \| \wt \beta_c - \beta^*  \|_{\Sigma}   \leq 4   \|   \nabla \wt \cL (\beta^* )  \|_{\Omega} \leq  4\left( C_2 r_0 \sqrt{\frac{p+u}{n}}     +   r^*\right)  \label{intermediate.bound}
\#
with probability at least $1-3e^{-u}$. Provided that the sample size is sufficiently large---$n\gtrsim p+u$, the right-hand side of the above inequality is strictly less than  $r_{\loc}= \kappa/(4\eta_{0.25})$ with $\kappa \asymp \sigma \sqrt{n/(p+u)}$. As a result, the intermediate estimator $\wt \beta_c$ falls into the interior of $\Theta(r_{\loc})$ with high probability conditioned on $\cE_0(r_0) \cap \cE_*(r_*)$. Via proof by contradiction, we must have $\wt \beta \in \Theta(r_{\loc})$ and hence $\wt \beta = \wt \beta_c$; otherwise if $\wt \beta\notin \Theta(r_{\loc})$, we have demonstrated that $\wt \beta_c$ must lie on the boundary of $\Theta(r_{\loc})$, which is a contradiction.
Consequently, the bound \eqref{intermediate.bound} also applies to  $\wt \beta$, as claimed.

\medskip
\noindent
{\sc Proof of \eqref{one-step.bahadur}}. To establish the Bahadur representation, note that the random process $\Delta_1(\cdot)$ defined in \eqref{def:Delta} can be written as $\Delta_1(\beta)   = \Sigma^{-1/2}  \{ \nabla \wt \cL (\beta) - \nabla \wt \cL ( \beta^*)   \} - \Sigma^{1/2} ( \beta - \beta^* )$. Moreover, note that
\#
 \nabla \wt \cL ( \beta^*)  = \nabla \wh \cL_{1,\kappa}(\beta^*) - \nabla \wh \cL_{1,\kappa}(\wt \beta^{(0)}) + \nabla  \wh \cL_\tau(\wt \beta^{(0)}) -  \nabla  \wh \cL_\tau(\beta^*)  +  \nabla  \wh \cL_\tau( \beta^*)  , \nn
\#
which in turn implies
\#
 \|  \nabla \wt \cL ( \beta^*) -  \nabla \wh \cL_\tau (\beta^* )   \|_{\Omega} \leq  \| \Delta_1(\wt \beta^{(0)})  \|_2 + \| \Delta (\wt \beta^{(0)})  \|_2  .	\nn
\#
Recall that $\nabla \wt \cL(\wt \beta)= 0$, and by \eqref{intermediate.bound}, $\| \wt \beta - \beta^* \|_{\Sigma} \leq r_1 := 4 C_2 r_0 \sqrt{(p+u)/n} +4  r_*$ with high probability conditioned on $\cE_0(r_0) \cap \cE_*(r_*)$. For $r_0 \geq 8 r_*$, we have $r_1 \leq r_0/2 + r_0/2 = r_0$ as long as $n\gtrsim p+u$, and hence $\wt \beta \in \Theta(r_0)$. Applying the bounds in \eqref{diff.grad.unif.bound} again, we obtain that conditioned on $\cE_0(r_0)\cap \cE_*(r_*)$,
\#
	&  \| \Sigma^{1/2} ( \wt \beta -\beta^* )  + \Sigma^{-1/2} \nabla \wh \cL_\tau(\beta^*) \|_2  \nn \\
& = \|  \Delta_1(\wt \beta ) + \Sigma^{-1/2} \nabla \wt \cL(\beta^*)  -  \Sigma^{-1/2}  \nabla \wh \cL_\tau(\beta^* )   \|_2 \nn \\
& \leq    \| \Delta_1(\wt \beta )  \|_2 +  \| \Delta_1(\wt \beta^{(0)})   \|_2+  \| \Delta (\wt \beta^{(0)})   \|_2  \nn \\
& \leq 2 \sup_{\beta \in \Theta(r_0)}    \| \Delta_1(\beta)   \|_2 + \sup_{\beta \in \Theta(r_0)}   \| \Delta (\beta)  \|_2 \nn \\
& \lesssim \sqrt{\frac{p+u}{n}} \cdot r_0 \nn
\# 
with probability at least $1-3e^{-u}$.  This completes the proof. \qed

\subsection{Proof of Theorem~\ref{thm:multi-step}}

Given a sequence of iterates $\{\wt \beta^{(t)} \}_{t = 0, 1, \ldots , T}$, we define   ``good" events 
\#
	\cE_t( r_t ) = \big\{  \wt \beta^{(t)} \in \Theta(r_t) \big\} , \ \  t= 0 ,\ldots, T,  \label{good.events}
\#
for some sequence of radii $r_0 \geq  r_1 \geq \cdots \geq  r_T >0$ to be determined. Examine the proof of Theorem~\ref{thm:one-step}, we see that the statistical properties of $\wt \beta^{(t)}$ depends on both first-order and second-order information of the loss function $\wt \cL^{(t)}(\cdot)$, namely, the $\ell_2$-norm of the gradient $\nabla \wt \cL^{(t)} (\beta^*)$ and the (symmetrized) Bregman divergence of $\wt \cL^{(t)}(\cdot)$. For the former,  we have
\#
 \nabla \wt \cL^{(t)}(\beta^*) = \nabla \wh \cL_{1,\kappa}(\beta^*) - \nabla \wh \cL_{1,\kappa}(\wt \beta^{(t-1)} )  +  \nabla \wh \cL_\tau(\wt \beta^{ ( t- 1 ) }  )  . \label{grad.itert}
\#
Let $\Delta_1(\cdot) $ and  $\Delta(\cdot)$ be the random processes defined in \eqref{def:Delta}, and observe that $\Sigma^{-1/2}\nabla \wt \cL^{(t)}(\beta^*)  =   \Delta(\wt \beta^{(t-1)} ) - \Delta_1(\wt \beta^{(t-1)})   +  \Sigma^{-1/2} \nabla \wh \cL_\tau(\beta^*)$.   By the triangle inequality,
\#
   \| \nabla \wt \cL^{(t)}(\beta^*)    \|_{\Omega} \leq      \| \Delta(\wt \beta^{(t-1)} )   \|_2 +     \| \Delta_1(\wt \beta^{(t-1)})  \|_2   +   \|  \nabla \wh \cL_\tau(\beta^*)   \|_{\Omega} . \label{general.gradient.bound}
\#
On the other hand,  note that the shifted Huber losses $\wt \cL^{(t)}(\cdot)$ have the same Bregman divergence, denoted by 
\#
	\overbar D(\beta_1 , \beta_2) = \langle \nabla \wt \cL^{(t)} (\beta_1) -  \nabla \wt \cL^{(t)}(\beta_2) , \beta_1 - \beta_2 \rangle = \langle  \nabla \wt \cL_{1,\kappa}  (\beta_1) -  \nabla \wt \cL_{1,\kappa}  (\beta_2)  , \beta_1 - \beta_2 \rangle.   \nn
\#
Define the local radius $r_{\loc}= \kappa / (4 \eta_{0.25})$. Then, applying Lemma~\ref{lem:local.RSC} with $r=r_{\loc}$ yields that, with probability at least $1-e^{-u}$,
\#
\overbar D(\beta  , \beta^*) \geq \frac{1}{4} \| \beta - \beta^* \|_{\Sigma}^2  \label{local.strong.convexity}
\#
holds uniformly over $\beta \in \Theta(r_{\loc})$. Let $\cE_{\lsc}$ be the event that the local strong convexity \eqref{local.strong.convexity} holds.

With the above preparations, we are ready to extend the argument in the proof of Theorem~\ref{thm:one-step} to deal with $\wt \beta^{(t)}$ sequentially. At each iteration, we construct an intermediate estimator $\wt \beta^{(t)}_{\imd}$---a convex combination of $\wt \beta^{(t)}$ and $\beta^*$---which falls in $\Theta(r_{\loc})$ and satisfies
\$
 \overbar D(\wt \beta_{\imd}^{(t)} , \beta^* ) \leq \| \nabla \wt \cL^{(t)} (\beta^*) \|_{\Omega} \cdot \| \wt \beta^{(t)}_{\imd} - \beta^* \|_{\Sigma} .
\$
If event $\cE_*(r_*)\cap \cE_{\lsc}$ occurs, the bounds \eqref{general.gradient.bound} and \eqref{local.strong.convexity} imply
\#
  \| \wt \beta^{(t)}_{\imd} - \beta^* \|_{\Sigma} \leq  4 \big\{  \| \Delta_1 (\wt \beta^{(t-1)} ) \|_2  + \| \Delta (\wt \beta^{(t-1)} )  \|_2 \big\}   +   4 r_* . \label{general.rate}
\# 
Moreover, it follows from \eqref{grad.itert} and the first-order condition $\nabla \wt \cL^{(t)} (\wt \beta^{(t)}) = 0$ that
\#
      & \| \Sigma^{1/2 } ( \wt \beta^{(t)} - \beta^* )  + \Sigma^{-1/2} \nabla \wh \cL_\tau(\beta^* ) \|_2  \nn \\
      & = \| \Sigma^{-1/2}  \{ \nabla \wt \cL^{(t)} (\wt \beta^{(t)})  -  \nabla \wt \cL^{(t)} (\beta^* ) \} - \Sigma^{1/2 } ( \wt \beta^{(t)} - \beta^* )     +  \Sigma^{-1/2} \{  \nabla \wt \cL^{(t)}(\beta^*) - \nabla \wh  \cL_\tau(\beta^* ) \} \|_2  \nn \\
      & \leq \| \Delta_1(\wt \beta^{(t)} ) \|_2 + \| \Delta_1(\wt \beta^{(t-1)}) \|_2 + \| \Delta(\wt \beta^{(t-1)} ) \|_2 . \label{general.bahadur}
\#
In view of the bounds in \eqref{diff.grad.unif.bound}, for every $0< r\leq \sigma$ we define the event
\#
	\cF(r) = \left\{ \sup_{\beta \in \Theta(r)} \big\{  \| \Delta_1 ( \beta) \|_2  + \| \Delta ( \beta)  \|_2 \big\}  \leq \gamma(u) \cdot r \right\}  \label{eventF}
\#
with $\gamma(u) = C \sqrt{(p+u)/n}$ for some $C>0$, which satisfies $\pr \{ \cF(r) \} \geq 1-2e^{-u}$.

Let $ 8r^* \leq  r_0 \leq \sigma$. In the following, we assume the event $\cE_0(r_0) \cap \cE_*(r_*) \cap \cE_{\lsc}$ occurs, and deal with $\{(\wt \beta^{(t)}_{\imd}, \wt \beta^{(t)} ), t= 1, 2, \ldots  , T\}$ sequentially. At iteration 1, it follows from \eqref{general.rate} that, conditioned on $\cF(r_0)$,
\#
 \| \wt \beta^{(1)}_{\imd} - \beta^* \|_{\Sigma} \leq  r_1 := 4 \gamma(u) \cdot r_0 + 4 r_* . \nn
\#
Provided that $n\gtrsim p+u$, we have $4\gamma(u)\leq 1/2 < 1$ and $r_1 \leq r_0 < r_{\loc} = \kappa / (4\eta_{0.25})$, so that $\wt \beta^{(1)}_{\imd} \in \Theta(r_1) \subseteq {\rm int}( \Theta(r_{\loc}))$. Via proof by contradiction, we must have $\wt \beta^{(1)} = \wt \beta^{(1)}_{\imd} \in \Theta(r_{\loc})$, which in turns certifies event $\cE(r_1)$.
Combining this with \eqref{general.bahadur}, we see that conditioned on $\cF(r_0)$, the event $\cE_1(r_1)$ mush happen and hence
\#
\left\{\begin{array}{ll}
 \| \, \wt \beta^{(1)}  - \beta^*\,  \|_{\Sigma}  \leq r_1 =4 \gamma(u) \cdot r_0 + 4 r_* \leq r_0  ,  \vspace{0.2cm}\\
  \| \,  \wt \beta^{(1)} - \beta^*   + \Sigma^{-1} \nabla \wh \cL_\tau(\beta^* ) \, \|_{\Sigma} \leq  2 \gamma(u)\cdot r_0 .
\end{array}\right.  \nn
\#

Now assume that for some $t\geq 1$, $\wt \beta^{(t)} \in \Theta(r_t)$ with $r_t = 4 \gamma(u) \cdot r_{t-1} + 4 r_* \leq r_{t-1} < r_{\loc}$. At $(t+1)$-th iteration, applying \eqref{general.rate} again yields that, conditioned on event $\cE_t(r_t) \cap \cF(r_t)$,
\$
 \| \wt \beta^{(t+1)}_{\imd} - \beta^* \|_{\Sigma} \leq r_{t+1} :=  4 \gamma(u) \cdot r_t + 4 r_* .
\$
By induction, $r_t\leq r_{t-1} < r_{\loc}$ so that $r_{t+1} \leq  4 \gamma(u) \cdot r_{t-1}+ 4 r_* = r_t < r_{\loc}$. This implies that $\wt \beta^{(t+1)}_{\imd}$ falls into the interior of $\Theta(r_{\loc})$, which enforces $\wt \beta^{(t+1)} = \wt \beta^{(t+1)}_{\imd} \in \Theta(r_{t+1}) $ and thus certifies event $\cE_{t+1}(r_{t+1})$. Combining this with the bound \eqref{general.bahadur}, we find that
\#
\left\{\begin{array}{ll}
 \|  \, \wt \beta^{(t+1)}  - \beta^* \, \|_{\Sigma}  \leq r_{t+1} =4 \gamma(u) \cdot r_t + 4 r_* \leq r_t  ,  \vspace{0.2cm}\\
  \| \, \wt \beta^{(t+1)} - \beta^*   + \Sigma^{-1} \nabla \wh \cL_\tau(\beta^* ) \, \|_{\Sigma} \leq  2 \gamma(u)\cdot r_t .
\end{array}\right.  \nn
\#

Repeat the above argument until we obtain $\wt \beta^{(T)}$. We have shown that conditioned on $ \cE_*(r_*) \cap \cE_{\lsc} \cap \cE_{t-1}(r_{t-1}) \cap \cF(r_{t-1})$ for every $0\leq t\leq T-1$, the event $\cE_t(r_t)$ must occur. Therefore, conditioned on $ \cE_*(r_*) \cap \cE_{\lsc} \cap \cE_0(r_0) \cap \{ \cap_{t=0}^{T-1} \cF(r_t) \}$, $\wt \beta^{(T)}$ satisfies the bounds
\#
\left\{\begin{array}{ll}
 \|  \, \wt \beta^{(T)}  - \beta^* \, \|_{\Sigma}  \leq r_T =4 \gamma(u) \cdot r_{T-1} + 4 r_*   ,  \vspace{0.2cm}\\
  \| \, \wt \beta^{(T)} - \beta^*   + \Sigma^{-1} \nabla \wh \cL_\tau(\beta^* ) \, \|_{\Sigma} \leq  2 \gamma(u)\cdot r_{T-1} .
\end{array}\right. \label{bound.T}
\#
Observe that $r_t = \{ 4\gamma(u) \}^t r_0 + \frac{1-\{4\gamma(u) \}^t}{1-4\gamma(u) } 4 r_*$ for $t=1,\ldots, T$. We choose $T$ to be the smallest integer such that $ \{ 4\gamma(u) \}^{T-1} r_0 \leq r_*$, that is, $T= \lceil \log(r_0/r_*) / \log(1/\{ 4\gamma(u)\}) \rceil + 1$. Consequently, the bounds in \eqref{bound.T} become
\#
\left\{\begin{array}{ll}
 \|  \, \wt \beta^{(T)}  - \beta^* \, \|_{\Sigma}  \leq    \big\{  \gamma(u)  + \frac{1}{1-4\gamma(u) } \big\} 4 r_*  \leq \{ 4 \gamma(u) + 8\}  r_*  ,  \vspace{0.2cm}\\
  \| \, \wt \beta^{(T)} - \beta^*   + \Sigma^{-1} \nabla \wh \cL_\tau(\beta^* ) \, \|_{\Sigma} \leq  18 \gamma(u)\cdot r_* .
\end{array}\right. \label{bound.T2}
\#

Finally, it suffices to show that the event $\cE_{\lsc}  \cap \{ \cap_{t=0}^{T-1} \cF(r_t) \}$ occurs with high probability. Recall from \eqref{local.strong.convexity} and \eqref{eventF}  that $\pr(\cE_{\lsc}) \geq 1-e^{-u}$ and $\pr\{\cF(r_t)\} \geq 1- 2e^{-u}$ for every $t=0,1,\ldots, T-1$. The claimed result then follows from \eqref{bound.T2} and the union bound.   \qed

\subsection{Proof of Theorem~\ref{thm:final.rate}}

Let $u>0$. Applying Theorem~B.1 in \cite{SZF2020} with a robustification parameter $\kappa \asymp \sigma \sqrt{n/(p+u)}$ yields that with probability at least $1-2e^{-u}$, $\| \wt \beta^{(0)} - \beta^*   \|_{\Sigma} \leq r_0 \asymp  \sigma \sqrt{(p+u)/n}$ as long as $n\gtrsim p + u$.
For event $\cE^*( r^*)$ defined in \eqref{event0},  we take $r^* \asymp \sigma\sqrt{(p+u)/N} + \tau (p+u)/N + \sigma^2/\tau $ in Lemma~\ref{lem:global.score} and obtain that   $\pr\{ \cE^*( r^*) \} \geq 1-e^{-u}$. Putting together the pieces, we conclude that event $\cE_0(r_0) \cap \cE_*(r_*)$ occurs with probability at least $1-3 e^{-u}$, provided that $n\gtrsim p+u$.

Set $u = \log n + \log_2 m$. Since $\tau \asymp \sigma \sqrt{N/(p+\log n + \log_2 m)}$, we see that
$$
	r_0 \asymp \sigma\sqrt{\frac{p+\log n + \log_2 m}{n}} ~~\mbox{ and }~~ r_* \asymp \sigma \sqrt{\frac{p+\log n + \log_2 m}{N}},
$$
and hence $r_0 / r_* \asymp \sqrt{m}$. 
Finally, applying Theorem~\ref{thm:multi-step}  yields the claimed bounds \eqref{final.rate} and \eqref{final.bahadur}. \qed

\subsection{Proof of Theorem~\ref{thm:CLT}}

For simplicity, we write $q= p + \log n + \log_2 m$ throughout the proof. For every vector $a\in \RR^p$,  define $S_{a} = N^{-1/2} \sum_{i=1}^N  \xi_i w_i$ and $S_{a}^0 = S_{a} - \EE  S_{a}$, where $\xi_i = \psi_\tau(\varepsilon_i)$ and $w_i =  a^\T \Sigma^{-1} x_i$. Under the moment condition $\EE (|\varepsilon|^{2+\delta} | x) \leq v_{2+\delta}$, using Markov's inequality yields $| \EE (\xi_i | x_i ) | \leq  \tau^{-1-\delta} \EE  (|\varepsilon_i|^{1+\delta}|x_i) \leq v_{2+\delta} \tau^{-1-\delta}$. Hence, $|\EE (\xi_i w_i)| \leq  v_{2+\delta} \|   a \|_{\Omega} \cdot \tau^{-1-\delta}$ and $|\EE  S_{a}| \leq v_{2+\delta} \| a \|_{\Omega} \cdot  N^{1/2}\tau^{-1-\delta}$.
 
With the above preparations, we are ready to prove the normal approximation for $\wt \beta$. Note that
\$
   & | N^{1/2} a^\T (\wt \beta - \beta^* ) - S^0_{a} |    \\ 
   & \leq   N^{1/2}   \left|  \left\langle  \Sigma^{-1/2} a, \Sigma^{1/2} (\wt \beta - \beta^* ) - \Sigma^{-1/2} \frac{1}{N} \sum_{i=1}^N \psi_\tau(\varepsilon_i) x_i \right\rangle  \right|  +  | \EE  S_{a} |  \\ 
   & \leq N^{1/2} \| a \|_{\Omega} \cdot \left\| \wt \beta - \beta^*  - \Sigma^{-1} \frac{1}{N} \sum_{i=1}^N \psi_\tau(\varepsilon_i) x_i \right\|_{\Sigma} + v_{2+\delta} \| a \|_{\Omega} \cdot  N^{1/2}\tau^{-1-\delta}.
\$
Applying \eqref{final.bahadur} in Theorem~\ref{thm:final.rate}, we find that with probability at least $1-C n^{-1}$,
\#
 | N^{1/2} a^\T (\wt \beta - \beta^* ) - S^0_{a} | \leq C_1 \| a \|_{\Omega} \cdot  \big(  \sigma  q n^{-1/2} + N^{1/2}  v_{2+\delta}\tau^{-1-\delta} \big) , \label{GAR1}
\#
where $C_1>0$ is a constant independent of $(N,n,p)$. 

For the centered partial sum $S^0_{a}$, it follows from the Berry-Esseen inequality (see, e.g. Theorem 2.1 in \cite{CS2001}) that
\#
 \sup_{t\in \RR } \big| \pr\big\{ S_{a}^0 \leq \var(S_{a}^0)^{1/2} t \big\} - \Phi(t) \big| \leq 4.1 \frac{  \EE  | \xi w - \EE (\xi w) |^{2+\delta} }{\var(\xi w )^{  1+\delta /2} N^{ \delta /2} },  \label{GAR2}
\#
where $\xi = \psi_\tau(\varepsilon)$ and $w = a^\T \Sigma^{-1} x$. Recall that $\tau \asymp \sigma \sqrt{N/q}$, and write $\sigma_{\tau, a}^2 = \EE (\xi w)^2 $.
By Proposition~A.2 in \cite{ZBFL2018}, $| \EE (\xi^2 | x)   -\sigma^2 | \leq 2\delta^{-1} v_{2+\delta}   \tau^{-\delta}  \asymp \delta^{-1} v_{2+\delta} \sigma^{-\delta} (q/N)^{\delta/2}$, and hence
\#
	\big|   \sigma_{\tau, a}^2/  (\sigma  \| a\|_{\Omega} )^2   - 1 \big| \lesssim \frac{v_{2+\delta}}{\delta \sigma^{2+\delta}} \bigg(\frac{q}{N} \bigg)^{\delta/2}. \label{var.compare.1}
\#
 Moreover, $\EE  |\xi w|^{2+\delta} \leq   \EE   |\varepsilon  w|^{2+\delta}   \leq  \mu_{2+\delta} \| a \|_{\Omega}^{2+\delta} v_{2+\delta}$, where $\mu_{2+\delta} := \sup_{u \in \mathbb{S}^{p-1}} \EE  |z^\T u |^{2+\delta}$ depends only on $\upsilon_1$ under Condition~(C1). Substituting these bounds into \eqref{GAR2} yields
\#
 \sup_{t\in \RR } \big| \pr\big\{ S_{a}^0 \leq \var(S_{a}^0)^{1/2} t \big\} - \Phi(t) \big| \leq  C_2     \frac{v_{2+\delta  }  }{\sigma^{2+\delta} N^{\delta/2}}, \label{GAR3}
\#
provided that $N\gtrsim q$. For the variance term, the bound $|\EE (\xi|x)| \leq \sigma^2 \tau^{-1}$ guarantees that
\$
	  \EE  (\xi w)^2 \geq \var(S_{a}^0) =    \EE  (\xi w)^2 - ( \EE \xi w )^2   \geq   \EE  (\xi w)^2 - ( \sigma \| a \|_{\Omega})^2 \cdot  \sigma^2\tau^{-2 } .
\$
Combined with \eqref{var.compare.1}, this implies $|  \var(S_{a}^0)  / \sigma_{\tau, a}^2 - 1   | \lesssim  \sigma^2 \tau^{-2}$,  from which it follows that
\#
 \sup_{t\in \RR } \big|  \Phi(t/\var(S_{a}^0)^{1/2}) - \Phi(t/ \sigma_{\tau, a} ) \big| \leq C_3  \frac{ \sigma^2 }{  \tau^2 }  . \label{GAR4}
\#

Let $G\sim \cN(0,1)$ and $t\in \RR$. Combining the bounds \eqref{GAR1}, \eqref{GAR3} and \eqref{GAR4}, we obtain
\$
 & \pr\big\{ N^{1/2} a^\T (\wt \beta - \beta^* ) \leq t\big\}  \\ 
& \leq \pr\big\{   S_{a}^0 \leq x + C_1   \| a \|_{\Omega} \cdot  \big(  \sigma q n^{-1/2} + N^{1/2} v_{2+\delta}\tau^{-1-\delta} \big) \big\} + C n^{-1} \\
& \leq \pr\big\{  \var(S^0_{a})^{1/2} G \leq t + C_1   \| a \|_{\Omega} \cdot  \big(  \sigma q n^{-1/2} + N^{1/2}  v_{2+\delta}\tau^{-1-\delta} \big) \big\} + C n ^{-1} +  C_2       \frac{v_{2+\delta  }  }{\sigma^{2+\delta} N^{\delta/2}} \\
& \leq \pr\big\{  \sigma_{\tau, a} G \leq  t + C_1  \| a \|_{\Omega} \cdot  \big( \sigma q n^{-1/2} + N^{1/2}  v_{2+\delta}\tau^{-1-\delta} \big) \big\}  +  C_2      \frac{v_{2+\delta  }  }{\sigma^{2+\delta} N^{\delta/2}} +  C_3   \frac{ \sigma^2 }{  \tau^2 }  \\
& \leq  \pr\big(  \sigma_{\tau, a} G \leq   t \big) + C n ^{-1} +  C_1 (2\pi)^{-1/2}   \big(   q n^{-1/2} + N^{1/2}  v_{2+\delta}\sigma^{-1}\tau^{-1-\delta} \big)    +  C_2       \frac{v_{2+\delta  }  }{\sigma^{2+\delta} N^{\delta/2}}+  C_3 \frac{ \sigma^2 }{  \tau^2 } .
\$
A similar argument leads to a series of reverse inequalities, and thus completes the proof. \qed

\subsection{Proof of Theorem~\ref{thm:hd.one-step}}

As before, we assume without loss of generality that $\cI_1 = \{1,\ldots, n\}$. Write $\wt \beta = \wt \beta^{(1)}$ for simplicity, and let $g = \wt \beta - \beta^*$ be the error vector. By the first-order optimality condition, there exists a subgradient $g \in \partial \| \wt \beta \|_1$ such that $g^\T \wt \beta = \| \wt \beta \|_1$ and $\nabla \wt \cL(\wt \beta) + \lambda \cdot g = 0$.
Moreover, the convexity of $\wt \cL(\cdot)$ implies
\$
	0 \leq  \overbar D_{\wt \cL}( \wt \beta ,\beta^*) = g^\T \big\{ \nabla \wt \cL(\wt \beta ) - \nabla \wt \cL(\beta^*) \big\} = - \lambda \cdot  g^\T g - g^\T \, \nabla \wt \cL(\beta^*) .
\$
Recall the true active set  $\cS = {\rm supp} (\beta^* ) \subseteq \{1,\ldots, p\}$, we have
\$
   -  g^\T g \leq \| \beta^* \|_1 - \| \wt \beta \|_1  = \| \beta^*_{\cS} \|_1 - \| g_{\cS^{{\rm c}}}  \|_1 - \| g_{\cS} + \beta_{\cS} \|_1 \leq \| g_{\cS} \|_1 - \| g_{\cS^{{\rm c}}} \|_1 .
\$
Together, the above two displays yield
\#
0 \leq  \overbar D_{\wt \cL}( \wt \beta ,\beta^*) \leq \lambda \bigl( \| g_{\cS} \|_1 - \| g_{\cS^{{\rm c}}} \|_1 \bigr) - g^\T \,\nabla \wt \cL(\beta^*). \label{hd1}
\#

To deal with $\nabla \wt \cL(\beta^*) = \nabla \wh  \cL_{1,\kappa}(\beta^*) - \nabla \wh \cL_{1,\kappa}(\wt \beta^{(0)} ) + \nabla \wh \cL_\tau( \wt \beta^{(0)})$, we define random processes
\$
\wh D_1( \beta )  = \nabla \wh \cL_{1,\kappa}(\beta) - \nabla \wh \cL_{1,\kappa}(\beta^*) ,  ~~~  \wh   D(\beta) = \nabla \wh \cL_\tau (\beta) - \nabla \wh \cL_\tau(\beta^*)   , 
\$
and write $D_1(\beta) = \EE  \wh D_1(\beta)$ and $D(\beta)=  \EE  \wh D(\beta) $. The gradient $\nabla \wt \cL(\beta^*)$ can thus be written as
\$
 \big\{ \wh D(\beta) - D(\beta) \big\} \Big|_{\beta = \wt \beta^{(0)}} +  \big\{ D_1(\beta) & - \wh D_1(\beta) \big\} \Big|_{\beta = \wt \beta^{(0)}}
 + \nabla \wh \cL_\tau(\beta^*) - \nabla \cL_\tau(\beta^*) \\
 & + \big\{  D(\beta) - D_1(\beta) \big\} \Big|_{\beta = \wt \beta^{(0)}} + \nabla \cL_\tau(\beta^*) . 
\$
For any $r>0$, define
\#
 \Delta_1(r) = \sup_{\beta \in \Theta(r) \cap \Lambda} \| \wh D_1(\beta) - D_1(\beta) \|_\infty, ~~ \Delta(r) = \sup_{\beta \in \Theta(r) \cap \Lambda} \| \wh D(\beta) - D(\beta) \|_\infty ,  \label{def:Deltar}  \\ 
 \delta(r) = \sup_{\beta \in \Theta(r)  } \|  D_1(\beta) - D (\beta) \|_{\Omega} ~\mbox{ and }~ b^* = \| \nabla \cL_\tau(\beta^* ) \|_{\Omega} . \label{def:deltar}
\#
The quantity $b^*$ can be viewed as the robustification bias and by Lemma~\ref{lem:global.score}, $b^*  \leq \sigma^2 \tau^{-1}$.

Back to the right-hand of \eqref{hd1}, conditioning on the event $\cE_0(r_0) \cap \cE_*(\lambda_*)$, it follows from H\"older's inequality that
\#
   |  g^\T \,\nabla \wt \cL(\beta^*)   |  \leq  \big\{  \Delta(r_0) +  \Delta_1(r_0) + \lambda_* \big\} \| g \|_1 + \{ \delta(r_0) + b^* \big\} \| g \|_{\Sigma} . \label{hd2}
\#
Let $\lambda = 2.5(\lambda_*+ \rho)$ for some $\rho>0$. Provided that
\#
 \rho \geq \max\big[   \Delta(r_0) +  \Delta_1(r_0)  ,  s^{-1/2}\big\{ \delta(r_0) + b^*  \big\} \big],  \label{rho.constraint}
\#
we have  $|  g^\T \,\nabla \wt \cL(\beta^*)   | \leq 0.4 \lambda \| g \|_1 + 0.4 s^{1/2} \lambda \| g \|_{\Sigma}$. Combined with \eqref{hd1}, this yields $
 0 \leq  1.4  \| g_{\cS} \|_1 - 0.6  \| g_{\cS^{{\rm c}}} \|_1   + 0.4 s^{1/2}  \| g \|_{\Sigma}$. Consequently, $ \| g \|_1 \leq  (10/3)  \| g_{\cS} \|_1 + (2/3) s^{1/2} \| g \|_{\Sigma}  \leq 4 s^{1/2} \| g \|_{\Sigma}$,
and hence $\wt \beta \in \Lambda$. Throughout the rest of the proof, we assume that the constraint \eqref{rho.constraint} holds.

Next, we apply Lemma~\ref{lem:local.RSC} to bound the left-hand side of \eqref{hd1} from below. As in the proof of Theorem~\ref{thm:one-step}, we set $r_{\loc}= \kappa/(4 \eta_{0.25})$ and define $\wt \beta_c = (1-c) \beta^* + c \wt \beta$, where $c= \sup\{u\in[0,1]: (1-u)\beta^* + u \wt \beta \in \Theta(r_{\loc})\}$. The same argument therein implies $\overbar D_{\wt \cL}(\wt \beta_c , \beta^*) \leq c \overbar D_{\wt \cL}(\wt \beta, \beta^*)$.
Recall that conditioned on $\cE_0(r_0) \cap \cE_*(\lambda_*)$, $\wt \beta$ falls in the $\ell_1$-cone $\Lambda$ and thus so does $\wt \beta_c$. Moreover, $\wt \beta_c \in \Theta(r_{\loc})$ by construction. Then it follows from Lemma~\ref{lem:local.RSC} that, with probability at least $1-e^{-u}$,
\#
  \overbar D_{\wt \cL}(\wt \beta_c , \beta^*)  \geq \frac{1}{4} \| \wt \beta_c - \beta^* \|_{\Sigma}^2  \nn
\#
as long as $n\gtrsim s \log p + u$. Combining this with \eqref{hd1}, \eqref{hd2} and \eqref{rho.constraint}, we obtain that
\#
\frac{1}{4} \| \wt \beta_c - \beta^* \|_{\Sigma}^2 \leq  c \lambda \big( 1.4 \| g_{\cS} \|_1 + 0.4 s^{1/2} \| g \|_{\Sigma} \big) \leq 1.8 s^{1/2} \lambda \| \wt \beta_c - \beta^* \|_{\Sigma} .  \nn
\#
Canceling $ \| \wt \beta_c - \beta^* \|_{\Sigma}$ on both sides yields
\# 
 \| \wt \beta_c - \beta^* \|_{\Sigma} \leq 7.2 s^{1/2} \lambda  . \label{hd3}
\#
Provided that $\kappa > 28.8 \eta_{0.25} s^{1/2} \lambda$, the right-hand side is strictly less than $r_{\loc}$. Via proof by contradiction, we must have $\wt \beta = \wt \beta_c \in \Theta(r_{\loc})$, and hence the bound \eqref{hd3} also applies to $\wt \beta$.

It remains to choose $\rho$ properly so that the constraint \eqref{rho.constraint} holds with high probability. Recall from Lemma~\ref{lem:global.score} that $b^* \leq \sigma^2 \tau^{-1}$.
 The following two lemmas provide upper bounds on the suprema $\Delta(r_0)$, $\Delta_1(r_0)$ and $\delta(r_0)$ defined in \eqref{def:Deltar} and \eqref{def:deltar}.

\begin{lemma} \label{lem:hd.grad}
Assume Condition~(C2) holds. Then, for any $r, u>0$,
\#
	 \Delta(r)  \leq   C_1  B^2  r  \Bigg\{   \sqrt{\frac{s\log(2p)}{N}} +   s^{1/2}\frac{\log(2p) + u}{N} \Bigg\} + C_2( \sigma_u \mu_4)^{1/2} r \sqrt{\frac{\log(2p)+u}{N}}
\#
with probability at least $1-e^{-u}$, where $C_1, C_2>0$ are absolute constants. The same bound, with $N$ replaced by $n$, holds for $\Delta_1(r)$.
\end{lemma}

\begin{lemma} \label{lem:deltar}
Condition~(C2) guarantees  $\delta(r ) \leq \kappa^{-2} r (\sigma^2 + \mu_4 r^2/3)$  for any $r>0$.
\end{lemma}

Let $0< r_0 \lesssim \sigma$ and set $\delta = 2e^{-u}$, so that $\log p + u \asymp \log(p/\delta)$. Suppose the sample size per machine satisfies $n\gtrsim s\log(p/\delta)$. Then, in view of Lemmas~\ref{lem:hd.grad} and \ref{lem:deltar}, a sufficiently large $\rho$, which is of order 
\$
 \rho \asymp  \max \Bigg\{   r_0 \sqrt{\frac{s \log(p/\delta)}{n}} , s^{-1/2}  \sigma^2  (  \kappa^{-2}  r_0  + \tau^{-1} ) \Bigg\},
\$
guarantees that \eqref{rho.constraint} holds  with probability at least $1-\delta /2$. With this choice of $\rho$, we see that the right-hand of \eqref{hd3} is strictly less than $r_{\loc}$ as long as $\kappa \gtrsim s^{1/2} \{ \lambda^* + r_0 \sqrt{ s\log(p/\delta)/n}  \} + \sigma^2 ( \kappa^{-2} r_0 +   \tau^{-1} )$. 
Since $\kappa \asymp  \sigma\sqrt{n/\log(p/\delta)}$, this holds trivially under the assumed sample size scaling, and thus completes the proof. \qed

 We end this subsection with the proofs of Lemmas~\ref{lem:hd.grad} and \ref{lem:deltar}.

\subsubsection{Proof of Lemma~\ref{lem:hd.grad}}

For any $r_1, r_2$, define the $\ell_1$/$\ell_2$-ball $\mathbb{B}(r_1,r_2)$. Consider the change of variable $v = \beta-\beta^*$, so that $v \in \mathbb B( 4s^{1/2}r , r )$ for $\beta \in \Theta(r) \cap \Lambda$. It follows that
\#
& \sup_{\beta \in \Theta(r) \cap \Lambda} \| \wh D(\beta) - D(\beta) \|_\infty \nn \\
&  \leq \max_{1\leq j\leq p} \sup_{v \in \mathbb B(4s^{1/2} r, r)}  \Bigg| \frac{1}{N}\sum_{i=1}^N (1-\EE )  \underbrace{ \big\{ \psi_\tau(\varepsilon_i - x_i^\T v ) - \psi_\tau(\varepsilon_i) \big\} x_{ij}  }_{=: \phi_{ij}(v)}\Bigg|  =  \max_{1\leq j\leq p} \Phi_j,
\#
where $\Phi_j := \sup_{v \in \mathbb B(4s^{1/2} r, r)}    | (1/N)\sum_{i=1}^N (1-\EE )\phi_{ij}(v)  |$ and $\psi_\tau(u) = \sign(u)\min(|u|,\tau)$. 
By the Lipschitz continuity of $\psi_\tau(\cdot)$, $ \sup_{v \in \mathbb B(4s^{1/2} r, r)}  |\phi_{ij}(v)| \leq  \sup_{v \in \mathbb B(4s^{1/2} r, r)} |x_{i}^\T v | \cdot  |x_{ij}| \leq 4 B^2 s^{1/2} r $ and, for each $v \in \mathbb B(4s^{1/2} r, r)$,
\$
 \EE  \phi_{ij}^2(v)   = \EE  \{  x_{ij}^2 (x_i^\T v)^2  \} \leq \big( \EE  x_{ij}^4 \big)^{1/2} \big\{ \EE  (x_i^\T v )^4 \big\}^{1/2} \leq \sigma_{jj} \mu_4\cdot r^2 .
\$
We then apply Bousquet's version of Talagrand's inequality \citep{B2003} and obtain that, for any $z>0$,
\# \label{Phij.concentration} 
\Phi_j  & \leq \EE  \Phi_j + \sup_{v \in \mathbb B(4s^{1/2} r, r)} \bigl\{ \EE  \phi_{ij}^2(v) \bigr\}^{1/2} \sqrt{\frac{2z}{N}} +4 \sqrt{  \EE  \Phi_j \cdot   B^2 s^{1/2} r \frac{z}{N} }      +  (4/3) B^2 s^{1/2} r   \frac{z}{  N}  \nn \\
& \leq   \EE  \Phi_j + (2 \sigma_{jj} \mu_4)^{1/2}  r \sqrt{\frac{ z}{N}} +4 \sqrt{  \EE  \Phi_j \cdot   B^2 s^{1/2} r \frac{z}{N} }      +  (4/3) B^2 s^{1/2} r   \frac{z}{  N} 
\#
with probability at least $1-2e^{-z}$. For  the expected value $ \EE  \Phi_j$, by Rademacher symmetrization we have
\#
 \EE  \Phi_j \leq 2  \EE   \sup_{v \in \mathbb B(4s^{1/2} r, r)}  \Bigg| \frac{1}{N}  \sum_{i=1}^N e_i \phi_{ij}(v) \Bigg| = 2 \EE \Bigg\{  \EE _e  \sup_{v \in \mathbb B(4s^{1/2} r, r)}  \Bigg| \frac{1}{N}  \sum_{i=1}^N e_i \phi_{ij}(v) \Bigg| \Bigg\} , \nn
\#
where $e_1, \ldots, e_N$ are independent Rademacher random variables. For each $i$, write $\phi_{ij}(v) = \phi_j(x_i^\T v)$, where $\phi_j(\cdot)$ is such that $\phi_j(0)=0$ and $|\phi_j(u) - \phi_j(v)|\leq |x_{ij}| \cdot | u-v| \leq B |u-v|$. It thus follows from Talagrand's contraction principle that 
\$
\EE _e  \sup_{v \in \mathbb B(4s^{1/2} r, r)}  \Bigg| \frac{1}{N}  \sum_{i=1}^N e_i \phi_{ij}(v) \Bigg| \leq 2 B\cdot  \EE _e  \sup_{v \in \mathbb B(4s^{1/2} r, r)}  \Bigg| \frac{1}{N}  \sum_{i=1}^N e_i  x_i^\T v  \Bigg| \leq 8 B s^{1/2} r \cdot \EE _e \Bigg\| \frac{1}{N}  \sum_{i=1}^N e_i x_i \Bigg\|_\infty. 
\$
Again, applying Hoeffding's moment inequality yields $\EE _e  \| (1/N)  \sum_{i=1}^N e_i x_i  \|_\infty \leq B \sqrt{2\log(2p)/N}$.  Putting together the pieces, we conclude that, for $j=1,\ldots, p$,
\#
\EE  \Phi_j \leq 16 B^2  r \sqrt{\frac{2s \log(2p)}{N}} . \nn
\#

Finally, taking $z= \log(2p)+u$ in \eqref{Phij.concentration}, the claimed bound follows from the union bound. \qed

\subsubsection{Proof of Lemma~\ref{lem:deltar}}

Let $\cL_\tau(\beta) = \EE  \wh \cL_\tau(\beta)$ be the population loss, so that 
\$
D_1(\beta) = \nabla \cL_\kappa(\beta) -\nabla \cL_\kappa(\beta^*) ~~\mbox{ and }~~ D(\beta) = \nabla \cL_\tau(\beta) - \nabla \cL_\tau(\beta^*). 
\$
Starting with $D_1(\beta)$, consider the change of variable $v = \Sigma^{1/2}(\beta -\beta^*)$. Then, by the mean value theorem for vector-valued functions,
\$
& \Sigma^{-1/2}D_1(\beta) - \Sigma^{ 1/2} (\beta - \beta^*)  \\
&= \Sigma^{-1/2} \EE  \int_0^1 \nabla^2 \cL_\tau\big( (1-t) \beta^* +  t \beta  \big)  {\rm d} t  \,\Sigma^{-1/2} \cdot v  - v \\
& = - \int_0^1 \EE  \big\{   \pr\big( | \varepsilon - t z^\T v | > \kappa | x \big)    z z^\T  \big\} {\rm d} t  \cdot v.
\$
Similarly, it can be obtained that
\$
 \Sigma^{-1/2}D (\beta) - \Sigma^{ 1/2} (\beta - \beta^*)   = - \int_0^1 \EE  \big\{   \pr\big( | \varepsilon - t z^\T v | > \tau  | x \big)    z z^\T  \big\} {\rm d} t  \cdot v.
\$
Recall that $\tau \geq \kappa >0$. We have 
\$
\Sigma^{-1/2}  \{ D_1(\beta)  -   D (\beta)  \} =  -\int_0^1 \EE  \big\{   \pr\big( \kappa< | \varepsilon - t z^\T v | \leq \tau   | x \big)    z z^\T  \big\} {\rm d} t  \cdot v
\$
By Markov's inequality and the fact that $\EE (\varepsilon | x)=0$, $\pr  ( | \varepsilon - t z^\T v | > \kappa  | x   )  \leq \kappa^{-2}  \{ \EE  (\varepsilon^2 | x) + t^2 (z^\T v)^2    \}\leq \kappa^{-2}  \{ \sigma^2+ t^2 (z^\T v)^2    \}$.
Substituting this into the above bound yields 
\$
 \sup_{\beta \in \Theta(r)}  \|   D_1(\beta) - D(\beta)  \|_{\Omega}    \leq  \kappa^{-2}   r  \bigl(  \sigma^2   + \mu_4  r^2 /3 \big) ,
\$
as desired. \qed

\subsection{Proof of Theorem~\ref{thm:hd.final}}

The proof will be carried out conditioning on the ``good event" $\cE_0(r_0) \cap \cE_*(\lambda_*)$ for some predetermined $0 < r_0, \lambda_* \lesssim \sigma$. 
Given $\delta \in (0,1)$, let the robustification parameters satisfy $\tau \geq  \kappa \asymp \sigma \sqrt{n/\log(p/\delta)}$.  Theorem~\ref{thm:hd.one-step} implies that the first iterate $\wt \beta^{(1)} \in \argmin_{\beta \in \RR^p} \{ \wt \cL^{(1)}(\beta) + \lambda_1 \| \beta \|_1 \}$ with 
\$
 \lambda_1 = 2.5(\lambda_* +  \rho_1) ~~\mbox{ and }~~ \rho_1 \asymp \max \Biggl\{ r_0 \sqrt{\frac{s \log(p/\delta)}{n}} ,  s^{-1/2}  \sigma^2  \tau^{-1} \Biggr\}
\$
satisfies the cone property $\wt \beta^{(1)} \in \Lambda $ and the error bound
\#
 \| \wt \beta^{(1)} - \beta^* \|_{\Sigma}  \leq   C_1  s \sqrt{\log(p/\delta) /n} \cdot r_0  + C_2 (\sigma^2 \tau^{-1} + s^{1/2} \lambda_* )  =: r_1  \label{iter1.ubd}
\#
with probability at least $1-\delta$. In \eqref{iter1.ubd}, we set $\alpha = \alpha(s,p,n,\delta) =  C_1  s \sqrt{\log(p/\delta) /n}$ and $r_* = C_2 (\sigma^2 \tau^{-1} + s^{1/2} \lambda_* )$, so that $r_1 = \alpha r_0 + r_*$. 
Provided the sample size per machine is sufficiently large, namely, $n\gtrsim s^2 \log(p/\delta)$, the contraction factor $\alpha$ is strictly less than 1, and hence the initial estimation error $r_0$ is reduced by a factor of $\alpha$ after one round of communication.

For $t=2,3,\ldots, T$, define the events $\cE_t(r_t) = \{ \wt \beta^{(t)} \in \Theta(r_t)  \cap \Lambda \}$ and radius parameters 
$$
	r_t = \alpha r_{t-1}+ r_* = \alpha^2 r_{t-2}+  (1+\alpha) r_*  = \cdots  =  \alpha^t  r_0 + \frac{1-\alpha^t}{1-\alpha } r_* .
$$
In the $t$-th iteration, we choose the regularization parameter $\lambda_t= 2.5 (\lambda_* +  \rho_t)$ with 
$$
\rho_t \asymp \max \Biggl\{ r_{t-1} \sqrt{\frac{s \log(p/\delta)}{n}} ,  s^{-1/2}  \sigma^2  \tau^{-1} \Biggr\} \asymp   s^{-1/2} \max \bigl\{  \alpha^t  r_0  ,    \sigma^2  \tau^{-1} \bigr\}  .
$$
Commenced with $\wt \beta^{(t-1)}$ at iteration $t\geq 2$, we apply Theorem~\ref{thm:hd.one-step} to obtain that conditioned on event $\cE_{t-1}(r_{t-1}) \cap \cE_*(\lambda_*)$, 
\#
  \wt \beta^{(t)} \in \Lambda  ~~\mbox{ and }~~ \| \wt \beta^{(t)} - \beta^* \|_{\Sigma} \leq  \alpha r_{t-1} + r_* = r_{t}  \label{itert.ubd}
\#
with probability at least $1-\delta$. In other words, event $\cE_t(r_t)$ occurs with probability at least $1-\delta$ conditioned on $\cE_{t-1}(r_{t-1}) \cap \cE_*(\lambda_*)$.

Finally, we choose $T =\lceil \log(r_0 / r_*) / \log(1/\alpha)  \rceil$ so that  $\alpha^T r_0 \leq r_*$. Then, applying \eqref{iter1.ubd}, \eqref{itert.ubd} and the union bound over $t=1,\ldots, T$ yields that, conditioned on $\cE_0(r_0)\cap \cE_*(r_*)$, the $T$-th iterate $\wt \beta^{(T)}$ falls into the cone $\Lambda$ and satisfies the error bound
\#
	\| \wt \beta^{(T)} - \beta^* \|_{\Sigma} \leq  r_T \asymp  r_* \nn
\#
with probability at least $1-T \delta$. This completes the proof of the theorem. \qed

\end{document}